\documentclass[10,journal]{IEEEtran}
\usepackage{amsmath,amssymb,amsfonts}
\usepackage{algorithmic}
\usepackage{algorithm}
\usepackage{array}
\usepackage[caption=false,font=normalsize,labelfont=sf,textfont=sf]{subfig}
\usepackage{textcomp}
\usepackage{stfloats}
\usepackage{url}
\usepackage{verbatim}
\usepackage{graphicx,psfrag,epsfig}
\usepackage{cite}
\usepackage{amssymb}
\usepackage{amsmath,lipsum}
\usepackage{cuted}
\usepackage{balance}
\usepackage{arydshln}
\usepackage{bm}
\usepackage{amsmath}
\usepackage{amsthm}
\usepackage{subeqnarray}
\usepackage{cases}
\usepackage{xcolor}
\usepackage{soul}
\sethlcolor{green}
\usepackage{xcolor}

\graphicspath{{./}}
\newtheorem{Theorem}{Theorem}
\newtheorem{Lemma}{Lemma}

\def\linspread{1}

\begin{document}


\title{A Covariance-Surrogate Trust-Region Framework for Movable-Antenna Enabled Anti-Jamming with Unknown Jammers}

\author{Lebin Chen, \IEEEmembership{Student Member,~IEEE,} Ming-Min Zhao,
	\IEEEmembership{Senior Member,~IEEE}, 
	Qingqing Wu, \IEEEmembership{Senior Member,~IEEE},
	 Min-Jian Zhao, \IEEEmembership{Member,~IEEE}, and Rui Zhang, \IEEEmembership{Fellow,~IEEE}

\thanks{{A preliminary version of this work has been accepted for
	 	presentation at the 2026 IEEE International Conference on
	 	Communications \cite{chen2026icc}.}}

	\thanks{L. Chen, M. M. Zhao, and M. J. Zhao are with the College of Information Science and Electronic
		Engineering, Zhejiang University, Hangzhou 310027, China, and also with the Zhejiang Provincial Key Laboratory of Multi‐Modal Communication Networks and Intelligent Information Processing, Hangzhou 310027, China (e-mail:
		12431101@zju.edu.cn; zmmblack@zju.edu.cn; mjzhao@zju.edu.cn).}
		\thanks{Qingqing Wu is with the Department of Electronic Engineering, Shanghai Jiao Tong University, Shanghai 200240, China (e-mail: qingqingwu@sjtu.edu.cn).}
	\thanks{
		R. Zhang is with 
		the Department of Electrical and Computer Engineering, National University of Singapore, Singapore 117583 (e-mail: elezhang@nus.edu.sg).
		}
}



\maketitle

\begin{abstract}
In this paper, we investigate a movable antenna (MA) enabled anti-jamming optimization problem, where a legitimate uplink system is exposed to multiple jammers with unknown jamming channels. 
To enhance the anti-jamming capability of the considered system, an  MA array is deployed at the receiver, and the antenna positions and the minimum-variance distortionless-response (MVDR) receive beamformer are jointly optimized to maximize the output signal-to-interference-plus-noise ratio (SINR). 
The main challenge arises from the fact that the interference covariance matrix is unknown and nonlinearly dependent on the antenna positions. 
To overcome these issues, we propose a surrogate objective by replacing the unknown covariance with the sample covariance evaluated at the current antenna position anchor. 
Under a two-timescale framework, the surrogate objective is updated once per block (contains multiple snapshots) at the current anchor position, while the MVDR beamformer is adapted on a per-snapshot basis. 
We establish a local bound on the discrepancy between the surrogate and the true objective by leveraging matrix concentration inequalities, and further prove that a natural historical-averaging surrogate suffers from a non-vanishing geometric bias.
Building on these insights, we develop a low-complexity projected trust-region (TR) surrogate optimization (PTRSO) algorithm that maintains the locality of each iteration via TR constraints and enforces feasibility through projection, which is guaranteed to converge to a stationary point of the surrogate problem near the anchor.
Numerical results verify the effectiveness and robustness of the proposed PTRSO algorithm, which consistently achieves higher output SINR than existing baselines.

\end{abstract}


\section{Introduction}
\IEEEPARstart{A}{ctive} jamming has become a central topic in wireless communications because it threatens the reliability of modern networks and safety-critical services \cite{AntiJammingSurvey,securitySurvey}. 
As spectrum reuse becomes denser and wireless links turn highly directional, even short bursts of targeted interference can severely disrupt communication, sensing, and control. 
Fixed-position antenna (FPA) arrays, such as uniform linear/planar arrays (ULAs/UPAs) and sparse arrays \cite{sparse}, are the workhorse of current communication, radar, and integrated sensing-and-communication deployments \cite{intro7,intro9,ap1}. 
In conventional FPA-based designs, physical-layer anti-jamming is typically pursued via spatial processing, which exploits the spatial degrees of freedom (DoFs) to differentiate the legitimate signal from jamming interference. A coherent strategy is to steer deep nulls toward the jammer's direction of arrival (DoA) while maintaining high gain in the desired direction \cite{Survey208,Survey227}. Techniques such as subspace projection based on multiple-input multiple-output (MIMO) processing \cite{Survey46,Survey47} and zero-forcing jamming suppression  mitigate interference by projecting the received signal onto a subspace orthogonal to the jamming channel \cite{Survey111}. 
However, in all these methods the array geometry is treated as fixed, so the achievable interference suppression capability is ultimately constrained by the static spatial response of the FPA.
In rich multipath environments, where spatial diversity is crucial for improving the receiver signal-to-noise ratio (SNR) and link reliability, FPAs are unable to adapt to substantial spatial and temporal variations of wireless channels and may therefore yield poor anti-jamming performance.

Fortunately, movable antenna (MA) systems, also known as fluid antenna systems (FASs) \cite{zhu2024historicalreviewfluidantenna}, offer a promising way to overcome FPA limitations. 
In an MA array, antenna elements are connected to RF chains through flexible links and can be physically repositioned in real time, thereby unlocking an additional layer of spatial reconfigurability on top of conventional beamforming. 
Under the field-response modeling framework, repositioning allows the receiver to reconfigure spatial responses, enhance desired-signal power, steer or carve spatial nulls, and flexibly shape beampatterns, which can be exploited for both sensing and communication tasks \cite{ma2024movableantennaenhancedwireless}. 
When the antenna-position vector is appropriately optimized, MA-enabled MIMO systems can effectively increase the usable aperture without adding hardware, leading to tangible capacity and SNR gains over their FPA counterparts \cite{10243545,MAsurvey,MA2}. 
These advantages have been demonstrated in diverse scenarios, including angle-of-arrival estimation and target detection \cite{ma2024movableantennaenhancedwireless,ISACMA}, joint position-and-beamforming optimization for ISAC \cite{ISACMA,FASforISAC1,R2,R3,R4}.
Such results collectively indicate that MA arrays are especially well suited to environments where propagation conditions, interference geometry, and performance objectives evolve on heterogeneous timescales.

Against this backdrop, research on MA-enabled anti-jamming is still in its infancy and, more importantly, that existing designs almost invariably rely on explicit information about the jammers. 
For instance, \cite{TVTantiJamming} assumed that the DoA information of both legitimate and jamming signals is available (often aided by learning-based tools), and then jointly optimized beamforming and antenna positions to suppress the jammers while preserving the desired signal. 
A complementary line of work adopted robust formulations in which each jammer was assumed to lie within a bounded angular-uncertainty set, and the MA geometry was optimized to guarantee performance against the worst-case realization in that set \cite{ICCantiJamming}. 
In parallel, \cite{ZeroForcing} showed that, by jointly optimizing the antenna-position vector and a zero-forcing (ZF) beamformer, an MA array can achieve full array gain for the desired user while steering deep spatial nulls toward interference directions. 
Nevertheless, these MA-aided anti-jamming and beamforming schemes still rely on substantial jammer-side prior information: they require either accurate DoAs (and often powers) of the jammers or well-calibrated angular uncertainty regions and quasi-static channel parameters, while ZF-based designs also demand sufficiently accurate channel state information at the receiver. 
In practical adversarial environments with noncooperative and possibly agile jammers, such detailed information may be unavailable or unreliable, which fundamentally limits the robustness of these approaches.

Motivated by this gap, and to the best of our knowledge, this work is the first to depart from such assumptions and study an MA-enabled receiver under unknown jamming channels, without requiring explicit jammer DoAs or parametric channel estimates. 
We consider an uplink communication system where an MA array is deployed at the receiver and the goal is to jointly optimize the antenna-position vector and the minimum variance distortionless response (MVDR) beamformer so as to maximize the output signal-to-interference-plus-noise ratio (SINR). 
The key difficulty lies in the fact that the interference-plus-noise covariance matrix is both unknown and nonlinearly dependent on the antenna positions, and can only be inferred from a finite number of snapshots, which inevitably introduces statistical uncertainty. 
The main contributions of this paper are summarized as follows:
\begin{itemize}
	\item 
	First, focus on the MA-enabled uplink system with unknown jammers, we formulate an anti-jamming optimization problem where the receiver jointly optimizes the MA positions and the MVDR receive beamformer to maximize the output SINR.
	The interference-plus-noise covariance seen by each antenna-position anchor is replaced by the corresponding blockwise sample covariance, which serves as a tractable surrogate for the true MVDR SINR objective.
	Moreover, by establishing the Lipschitz continuity of the array response and the covariance matrix with respect to antenna positions and invoking matrix concentration inequalities, we derive local error bounds that quantify how closely the surrogate objective tracks the true SINR as a function of the snapshot budget and the distance to the anchor.
	We further compare this local-anchor surrogate with a natural historical-averaging surrogate and prove that the latter introduces a non-vanishing geometric bias in the MVDR cost, thereby theoretically justifying the proposed surrogate design and clarifying the tradeoff between geometric adaptivity and statistical reliability.

	\item 
	Next, building on the above surrogate model, we propose a projected trust-region surrogate optimization (PTRSO) algorithm to optimize the MA positions.
	For each fixed covariance surrogate at a given anchor, PTRSO proceeds in iterations that construct a local quadratic model of the surrogate objective around the current MA configuration, compute a trust-region (TR) trial step by approximately solving the associated TR subproblem via Hessian-vector products and a Steihaug-conjugate-gradient procedure, and then project the trial point onto the MA-geometry constraint set.
	We show that PTRSO generates a sequence of iterates whose accumulation points satisfy first-order stationarity conditions of the surrogate problem, while its per-iteration complexity scales as $\mathcal{O}(N_r^2)$. This makes the algorithm suitable for large MA arrays and more favorable than naive line-search methods, such as projected gradient or Newton-type updates.

	\item 
	Finally, numerical results verify the effectiveness and robustness of the proposed PTRSO-based MA design.  
	Across a wide range of user SNRs, snapshot budgets, jammer numbers, and array sizes, PTRSO consistently achieves higher output SINR and more reliable anti-jamming performance than projected gradient, projected Newton, historical-average, and fixed-position array baselines.  
	The simulations corroborate the derived approximation bounds, reveal the impact of the snapshot budget and the anchor displacement on the surrogate fidelity, and highlight the importance of TR constraints and locality in reliably transferring surrogate gains into true SINR improvements.
\end{itemize}

The remainder of this paper is organized as follows. 
Section~II introduces the system model and problem formulation.
Section~III develops the theoretical properties of the proposed covariance surrogate and compares the local-anchor surrogate with a historical-averaging alternative. 
Section~IV presents the PTRSO algorithm together with its convergence property and an intuitive analysis of its behavior. 
Section V provides numerical results. Finally, conclusions
are drawn in Section VI.

\textit{Notations:}
Scalars, vectors, and matrices are denoted by lower/upper case, bold-face lower-case, and bold-face uppercase letters, respectively.
The symbol $j$ denotes the imaginary unit, i.e., $j^2=-1$.
$(\cdot)^T$, $(\cdot)^*$ and $(\cdot)^H$ denote the transpose, conjugate and conjugate transpose operators, respectively. 
The inner product between two complex vectors $\mathbf a$ and $\mathbf b$ is denoted by
$\langle \mathbf a,\mathbf b\rangle \triangleq \mathbf a^{H}\mathbf b$.
We use $\Re\{x\},\Im\{x\}$ and $\angle x $ to denote the real part, imaginary part and argument of a complex number $x$. $\mathbf{I}$, $\mathbf{0}$ and $\mathbf{1}$  are used to represent identity matrix, all-zero vector and all-one vector with proper dimensions, respectively. 
$\|\cdot\|$ denotes the Euclidean norm of a complex vector
and the spectral norm of a complex matrix, $\|\cdot\|_{\mathrm F}$ denotes the
Frobenius norm of a matrix, $\|\cdot\|_{\psi_2}$ denotes the sub-Gaussian
norm of a real-valued random variable (i.e., $\|a\|_{\psi_2}
= \inf\{t >0:\mathbb{E}\{e^{a^2/t^2}\}\le2  \}$), $\|\cdot\|_{L_2}$
denotes the $L^2$-norm of a random variable (i.e., $\|a\|_{L_2}^2
= \mathbb{E}\{|a|^2\}$), and $|\cdot|$ denotes the absolute value of a
complex number.
For a scalar argument $t\ge 0$,
$f(t) = \mathcal{O}(\phi(t))$ as $t \to 0$ means that 
$\lvert f(t)/\phi(t)\rvert$ remains bounded for all sufficiently small $t$.
The set of integers is denoted by $\mathbb{Z}$, and the sets of $P \times Q$ dimensional complex, real and positive real matrices are denoted by $\mathbb{C}^{P\times Q}$, $\mathbb{R}^{P\times Q}$ and $\mathbb{R}^{P\times Q}_{++}$, respectively.

\begin{figure}[t]
	\centering
	\includegraphics[width=2.4in]{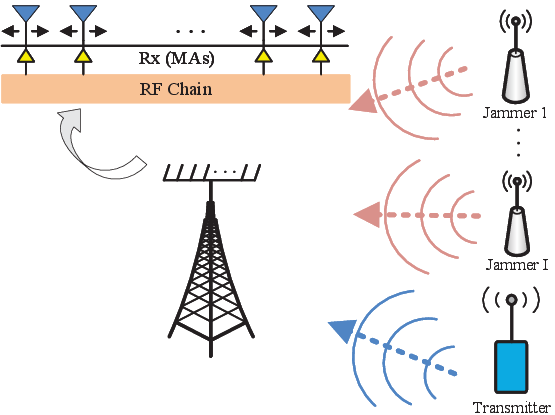}
	\caption{Illustration of MA-enabled MIMO receiver under jamming attacks.}
	\label{model}
	\vspace{-0cm}
\end{figure}

\section{System Model And Problem Formulation}
\subsection{System Setting}
In this paper, we consider an MA-enabled anti-jamming communication system in the presence of \(I\) unknown jammers, and the receiver is equipped with \(N_r\) MAs collecting \((I\!+\!1)\) far-field uncorrelated signals \(\{s_i\}_{i=0}^{I}\), where \(s_0\) is the desired communication signal from a single-antenna transmitter, and \(\{s_i\}_{i=1}^{I}\) are interference signals generated by the jammers, as depicted in Fig. \ref{model}.
{For practicality, we assume that the jammers and the receiver are {noncooperative}, thus the jammer-side channel information is unavailable to the receiver.
	To cope with this lack of prior knowledge, the receive MAs can be dynamically repositioned in real-time through flexible cables linked to the RF chains to exploit its spatial reconfigurability\cite{MAsurvey}.}
Furthermore, we assume narrow-band quasi-static channels, such that the time overhead for adjusting MA positions is tolerable compared to the much longer channel coherence time of the transmitter\cite{MAsurvey}. More specifically, we consider a linear MA array of size $N_r$ at the receiver and denote $\mathbf{x}=[x_1,x_2,...,x_{N_r}]^T\in\mathbb{R}^{N_r}$ as the corresponding antenna position vector (APV). The corresponding steering vector can be written as
\begin{equation}
	\label{eq:a_theta}
	\mathbf{a}(\mathbf{x},\theta)=\begin{bmatrix}e^{jkx_1\sin\theta},e^{jkx_2\sin\theta},...,e^{jkx_{N_r}\sin\theta}\end{bmatrix}^T,
\end{equation}
where $k=\frac{2\pi}{\lambda}$, $\lambda$ denotes the carrier wavelength and $\theta$ is the azimuth angle of arrival (AoA).

In the considered system, we adopt a field-response-based channel model \cite{10243545}, where the number of communication paths is denoted as $L$, the azimuth angle of departure (AoD) of the $\ell$-th ($\ell=1,2,...,L$) path is denoted by $\theta_\ell\in[0,2\pi)$, and the channel gain vector is defined as $\bm{\alpha}\in\mathbb{C}^{L\times1}$. 
Thus, the communication channel vector between the legitimate transmitter and receiver is given by 
\begin{equation}
	\label{eq:h0}
	\mathbf{h}_0(\mathbf{x})= [\mathbf{a}(\mathbf{x},\theta_1),\mathbf{a}(\mathbf{x},\theta_2),...,\mathbf{a}(\mathbf{x},\theta_L)] \bm{\alpha}\triangleq
	\mathbf{A}(\mathbf{x})\boldsymbol{\alpha},
\end{equation}
which is assumed to be known to the receiver \cite{hybridAnti,10243545}.
The signals collected by all receiving antennas are processed by a digital beamformer $\mathbf{w}\in\mathbb{C}^{{{N_r}}\times 1}$, then the output signal can be expressed as
\begin{equation} 
	\label{eq:hbf_out}
	y(\mathbf{x},t)=\mathbf{w}^{H}\mathbf{r}(\mathbf{x},t),
\end{equation}
where $t$ represents the time index, and $\mathbf{r}(\mathbf{x},t)$ is the received signal at the $t$-th time instant which can be written as
\begin{equation} 
	\label{eq:r_t}
	\mathbf{r}(\mathbf{x},t)=\mathbf{h}_0(\mathbf{x}) s_0(t)+\sum_{i=1}^{I}\mathbf{g}_i(\mathbf{x}) s_i(t)+\mathbf{n}(t),
\end{equation}
$\mathbf{g}_i(\mathbf{x})\in\mathbb{C}^{N_r\times 1}$ represents the jamming channel vector between the $i$-th jammer and the receiver, and $\mathbf{n}(t)\in\mathbb{C}^{N_r\times 1}$ denotes the noise vector whose elements follow the complex Gaussian distribution $\mathcal{CN}(0,\sigma_n^2)$.
Since we assume narrow-band quasi-static channels, over each static time
interval the desired symbol process $\{s_0(t)\}$ can be modeled as
zero-mean and wide-sense stationary (WSS) with average power
$\sigma_s^{2}\triangleq\mathbb{E}\{|s_0(t)|^{2}\}$, and each jammer
symbol process $\{s_i(t)\}_{i=1}^{I}$ is also WSS
with constant second moment
$\sigma_i^{2}\triangleq\mathbb{E}\{|s_i(t)|^{2}\}$, $i=1,\ldots,I$.
The processes $\{s_0(t)\}$, $\{s_i(t)\}_{i=1}^I$, and
$\{\mathbf n(t)\}$ are mutually independent across $t$.
In particular, the communication signal is uncorrelated with the
interference and noise.

For the jamming model, it is assumed that the jammers can interfere the receiver from various angles. The jamming channel vector between the $i$-th jammer and the receiver is modeled as~\cite{hybridAnti}\footnote{Because the jammers and receiver are noncooperative, accurately estimating the jamming channels via pilots is difficult. We therefore treat the jamming channels as unknown. The parametric model (\ref{eq:gi}) here is only used for illustration purpose.}
\begin{equation}
	\label{eq:gi}
	\mathbf{g}_i(\mathbf{x})
	=\zeta_i\mathbf{a}(\mathbf{x},\phi_i),
\end{equation}
where $\zeta_i$ represents the complex channel gain,
and $\phi_i$ is the AoA.
The SINR after receive beamforming can thus be expressed as
\begin{equation}
	\label{eq:sinr}
	\begin{aligned}
		\mathrm{SINR}
		&=\frac{\mathbb{E}\left\{\left|\mathbf{w}^{H}\mathbf{h}_0(\mathbf{x})s_0(t)\right|^{2}\right\}}
		{\mathbb{E}\left\{\left|\mathbf{w}^{H}\left(\sum_{i=1}^{I}\mathbf{g}_i(\mathbf{x}) s_i(t)+\mathbf{n}(t)\right)\right|^{2}\right\}}\\
		&=\frac{\sigma_s^{2}\left|\mathbf{w}^{H}\mathbf{h}_0(\mathbf{x})\right|^{2}}
		{\mathbf{w}^{H}\mathbf{R}_{i+n}(\mathbf{x})\mathbf{w}},
	\end{aligned}
\end{equation}
where 
$\mathbf{R}_{i+n}(\mathbf{x})=\mathbb{E}\left\{\mathbf{p}(\mathbf{x},t)\mathbf{p}(\mathbf{x},t)^{H}\right\}$ represents the interference-plus-noise covariance matrix with $\mathbf{p}(\mathbf{x},t)=\sum_{i=1}^{I}\mathbf{g}_i(\mathbf{x}) s_i(t)+\mathbf{n}(t)$.
The covariance matrix of $\mathbf{r}(\mathbf{x},t)$ can be further expressed as
\begin{equation} 
	\label{eq:Rxx}
	\begin{aligned}
		\mathbf{R}(\mathbf{x})&=\mathbb{E}\left\{\mathbf{r}(\mathbf{x},t)\mathbf{r}(\mathbf{x},t)^{H}\right\}\\
		&=\sigma_s^{2}\mathbf{h}_0(\mathbf{x})\mathbf{h}_0(\mathbf{x})^{H}
		+\sum_{i=1}^{I}\sigma_i^{2}\mathbf{g}_i(\mathbf{x})\mathbf{g}_i(\mathbf{x})^{H}
		+\sigma_n^{2}\mathbf{I}_{N_r}.
	\end{aligned}
\end{equation}

\subsection{Problem Formulation}
In this paper, we consider a joint receive beamformer and antenna position optimization problem in the presence of multiple jammers under the MVDR framework \cite{AdaptiveArrays}.
In particular, we aim to maximize the output SINR, or equivalently minimize the output variance $\mathbf{w}^{H} \mathbf{R}(\mathbf{x})\mathbf{w}$,
subject to a {unit-gain distortionless constraint} for the desired signal and the array geometry constraints, i.e.,
\begin{equation}
	\label{eq:8}
	\begin{aligned}
		\min_{\mathbf{w},\mathbf{x}}\quad 
		& \mathbf{w}^{H} \mathbf{R}(\mathbf{x})\mathbf{w} \\
		\text{s.t.}\quad 
		& \left|\mathbf{w}^{H}\mathbf{h}_0(\mathbf{x})\right|=1,\\
		&\begin{vmatrix}
			x_m-x_n
		\end{vmatrix}\ge{d}, 0\le x_i \le D_x, \\ &1\le{m},{n},{i}\le{N_r},m\neq{n},
	\end{aligned}
\end{equation}
where the first constraint ensures that the communication signal passes through without distortion, $d$ in the second constraint is the minimum distance (usually we set $d=\frac{\lambda}{2}$) between any two MAs to avoid the coupling effect, and $D_x$ is the aperture of the overall receive antenna array.\footnote{{{For the single desired stream considered herein, the linear minimum mean-square error (LMMSE)
	and MVDR receive beamformers have the same direction. In particular,
	the LMMSE beamformer is proportional to
	$\big(
	\sigma_s^2\mathbf h_0(\mathbf x)\mathbf h_0^H(\mathbf x)
	+\mathbf R_{i+n}(\mathbf x)
	\big)^{-1}\mathbf h_0(\mathbf x)$, which is further proportional to
	$\mathbf R_{i+n}^{-1}(\mathbf x)\mathbf h_0(\mathbf x)$ by the matrix
	inversion lemma. Hence, they achieve the same output SINR, which is
	invariant to nonzero scaling. We adopt the MVDR
	normalization to impose the unit-response constraint
	$|\mathbf w^H\mathbf h_0(\mathbf x)|=1$.}}}
Due to the equivalence between the antennas, we assume $x_1<x_2<...<x_{N_r}$  without loss of optimality. 
Thus, the second constraint in~\eqref{eq:8} can be equivalently transformed into a linear inequality constraint $\mathbf{U}\mathbf{x}\preceq\mathbf{l}$ and the corresponding feasible set for $\mathbf x$ can be defined as $\mathcal X \triangleq \{\mathbf x: \mathbf{U}\mathbf{x}\preceq\mathbf{l}\}$, where 
\begin{equation}
	\label{q16}
	\begin{aligned}
		\mathbf{U}&=\left[\begin{array} {ccccccc}
			1 & -1 & 0 & \cdots & 0 & 0\\
			0 & 1 & -1 & \cdots & 0 & 0\\
			\vdots & \vdots & \vdots & \ddots & \vdots & \vdots\\
			0 & 0 & 0 & \cdots & 1 & -1\\
			-1 & 0 & 0 & \cdots & 0 & 0\\
			0 & 0 & 0 & \cdots & 0 & 1
		\end{array}\right]_{(N_r+1)\times{N_r}},
	\end{aligned}
\end{equation}
\begin{equation}
	\begin{aligned}
		\mathbf{l}&=\left[
		-d,  -d ,\cdots , -d ,0, D_x
		\right]^T_{(N_r+1)\times1}.
	\end{aligned}
\end{equation}

With given $\mathbf x$, the following well-known solution can be found for the optimal beamforming vector \cite{AdaptiveArrays}:\footnote{Equation~(\ref{wopt}) is the oracle MVDR beamformer that assumes perfect knowledge of the true covariance $\mathbf R(\mathbf x)$ and is used only to eliminate $\mathbf w$ and derive the antenna-position-only objective in (\ref{P}). In practice, $\mathbf R(\mathbf x)$ is unavailable, and thus the implemented beamformer is computed from an estimated surrogate covariance $\hat{\mathbf R}$ (e.g., the blockwise sample covariance), as described below.}
\begin{equation}
	\label{wopt}
	\begin{aligned}
		\mathbf{w}_\text{opt}(\mathbf x)=\frac{\mathbf{R}^{-1}(\mathbf{x})\mathbf{h}_0(\mathbf{x})}{\mathbf{h}_0^H(\mathbf{x}) \mathbf{R}^{-1}(\mathbf{x})\mathbf{h}_0(\mathbf{x})}.
	\end{aligned}
\end{equation}
Consequently, problem \eqref{eq:8} can be equivalently transformed into the following problem after adopting \eqref{wopt} without loss of optimality:
\begin{equation}
	\label{P}
	\begin{aligned}
		\max_{\mathbf{x} \in \mathcal X}\quad 
		{g}(\mathbf{x})\triangleq\mathbf{h}_0^H(\mathbf{x}) \mathbf{R}^{-1}(\mathbf{x})\mathbf{h}_0(\mathbf{x}).
	\end{aligned}
\end{equation}

\begin{figure}[t]
	\centering
	\includegraphics[width=3.5in]{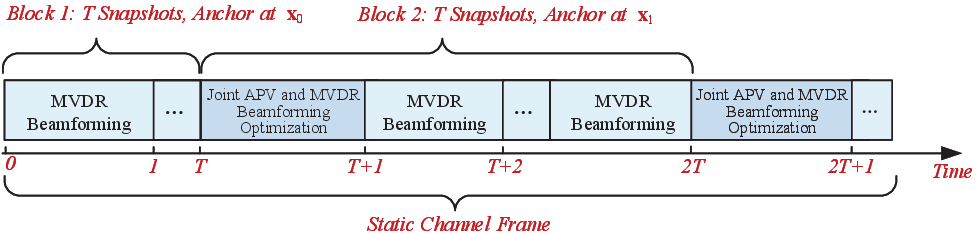}\vspace{-0.0cm}
	\caption{Illustration of the proposed two-timescale framework.}
	\label{twotimescale}
	\vspace{-0.0cm}
\end{figure}

Notably, the covariance matrix $\mathbf{R}(\mathbf{x})$ is unavailable in practice and is typically estimated from a limited number of snapshots. 
For practical considerations, e.g., reducing computational complexity and optimization overhead, we adopt a two-timescale optimization strategy, where the beamforming vector $\mathbf w$ is updated on a per-snapshot basis using a blockwise sample covariance estimate within the current block, while the antenna position $\mathbf x$ is optimized less frequently (i.e., once per block of $T$ snapshots) to cater to the slower MA actuation. 
Therefore, as illustrated in Fig.~\ref{twotimescale}, we partition the static channel time frame into blocks of 
$T$ snapshots. Within each block, the antenna positions are fixed at the current configuration, which we denote as the {antenna position anchor}, e.g., \(\mathbf x_0\).
In the proposed two-timescale design, the antenna positions are kept fixed at the current anchor \(\mathbf x_i\) over a block of \(T\) snapshots, based on which we can obtain a blockwise sample covariance at \(\mathbf x_i\) as
\begin{equation}
	\label{eq:9}
	\hat{\mathbf{R}}(\mathbf{x}_i)
	=\frac{1}{T}\sum_{t=1}^{T}\mathbf{r}(\mathbf{x}_i,t)\mathbf{r}^{H}(\mathbf{x}_i,t).
\end{equation}
Then, by replacing the unknown interference covariance matrix $\mathbf R(\mathbf x)$ in (\ref{P}) by \(\hat{\mathbf R}(\mathbf x_i)\), a block-level surrogate objective can be acquired for updating the antenna positions, and the corresponding problem can be rewritten as follows:
\begin{equation}
	\label{P1}
	\begin{aligned}
		\max_{\mathbf{x}\in \mathcal X}\quad 
		\hat{g}(\mathbf{x})\triangleq\mathbf{h}_0^H(\mathbf{x}) \hat{\mathbf{R}}^{-1}(\mathbf{x}_i)\mathbf{h}_0(\mathbf{x}).
	\end{aligned}
\end{equation}
Although replacing \(\mathbf R(\mathbf x)\) with the blockwise sample covariance \(\hat{\mathbf R}(\mathbf x_i)\) at the current anchor $\mathbf x_i$ may look trivial as the true covariance \(\mathbf R(\mathbf x)\)  varies with the antenna positions, we can justify that increasing the surrogate objective \(\hat g(\mathbf x)\) also improves the true objective \(g(\mathbf x)\) within a neighborhood of the anchor $\mathbf x_i$. In the next subsection, we will establish a quantitative local error bound for \(\hat g(\mathbf x)\), and show that the mismatch \(|g(\mathbf x)-\hat g(\mathbf x)|\) is controlled by two terms, one proportional to the displacement from the anchor $\mathbf{x}_i$, i.e., \(\|\mathbf x-\mathbf x_i\|\), and another that decreases with the increasing of the snapshot budget \(T\). 
Consequently, keeping the updates close to the anchor $\mathbf{x}_i$ and increasing $T$ tighten this bound, thereby ensuring that the gains achieved on the surrogate objective reliably transfer to those on the true objective $g(\mathbf{x})$.\footnote{As an alternative, one may estimate the interference AoAs and powers per block (e.g., via parametric subspace methods) and then synthesize a parametric interference covariance. We instead adopt a nonparametric, blockwise sample-covariance surrogate here because (i) it avoids model-order selection and repeated spectral decompositions, reducing computational overhead; and (ii) it is empirically robust with finite snapshot number $T$ and already achieves high performance in our simulations. Parametric covariance reconstruction based on AoA/power estimation is a natural extension and will be pursued in our follow-up work.}

\begin{figure}[t]
	\centering
	\includegraphics[width=2.4in]{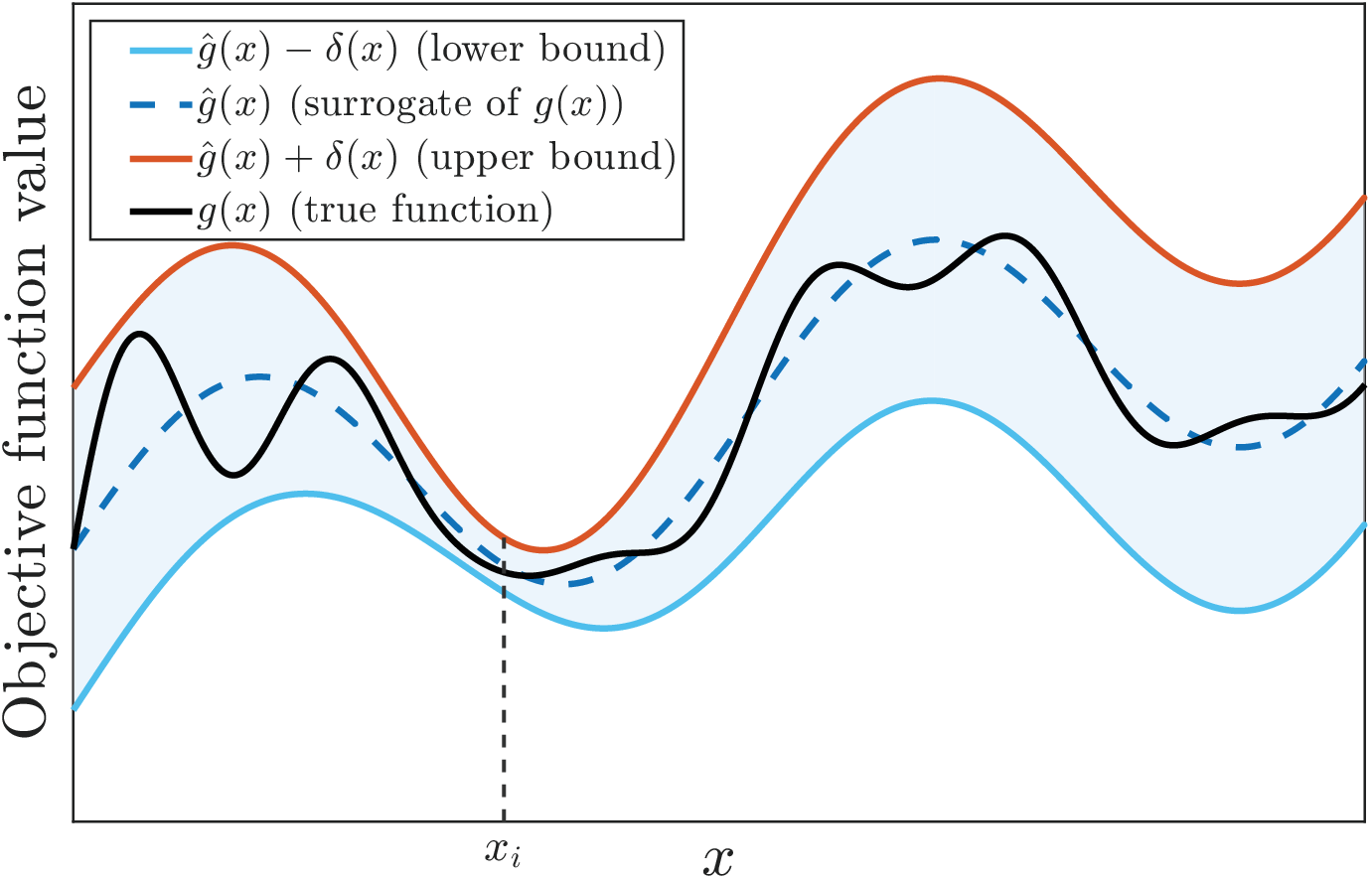}
	\caption{Toy example of surrogate and true function geometry.}
	\label{surragate}
	\vspace{-0cm}
\end{figure}

\section{Theoretical Properties of Covariance Surrogates}
In this section, we study the theoretical properties of the covariance
surrogates \(\hat{\mathbf R}(\mathbf x_i)\) that support our MA position optimization.
Our main objective is twofold: to quantify how accurately the local-anchor
surrogate can track the true SINR objective, and to reveal how this
accuracy scales with the snapshot budget and the displacement from the
anchor.
By combining the Lipschitz properties of the array response with matrix
concentration inequalities for the sample covariance, we derive local
approximation bounds that characterize the mismatch between the surrogate
objective and the true one.
We then compare the local-anchor model with a historical-average surrogate
that aggregates covariance information from multiple anchors and appears
attractive from a purely statistical perspective. However, we show that this
surrogate suffers from an intrinsic geometric bias, which makes it fundamentally
less accurate for TR updates.
These insights provide the theoretical foundation for the algorithmic design
developed in Section~IV.

\subsection{Theoretical Justification of the Proposed Surrogate }

To this end, we present several auxiliary lemmas to establish (i) the Lipschitz continuity
of the steering vectors and $\mathbf R(\mathbf x)$ with respect to $\mathbf x$, and
(ii) a concentration bound for the sample covariance $\hat{\mathbf R}(\mathbf{x})$ at $\mathbf x_i$.

\begin{Lemma}
	For any $\theta$ and $\mathbf x,\mathbf y\in\mathbb R^{N_r\times1}$, the following inequality holds:
	\begin{equation}
		\label{eq:a-Lip-short}
		\|\mathbf a(\mathbf x,\theta)-\mathbf a(\mathbf y,\theta)\|
		\le k|\sin\theta|\,\|\mathbf x-\mathbf y\|.
	\end{equation}
\end{Lemma}
\begin{proof}
	Please refer to Appendix A. 
\end{proof}
\noindent
Building on this Lipschitz property of the steering vector, we next show
that the covariance matrix $\mathbf R(\mathbf x)$ is also globally
Lipschitz-continuous with respect to the APV.
\begin{Lemma}
	For all $\mathbf x,\mathbf y\in\mathbb R^{N_r\times1}$,
	\begin{equation}
		\label{eq:R-Lip-short}
		\|\mathbf R(\mathbf x)-\mathbf R(\mathbf y)\| \le L_R\|\mathbf x-\mathbf y\|
	\end{equation}
	holds, where 
	\begin{equation}
		L_R\!=\!{2k}{\sqrt{N_r}}\Big(\sigma_s^2\sqrt{L}\|\boldsymbol\alpha\|^2 \sqrt{\sum_{\ell=1}^{L}\sin^2\theta_\ell}+\sum_{i=1}^{I}\sigma_i^2|\zeta_i|^2|\sin\phi_i|\Big).
	\end{equation} 
\end{Lemma}
\begin{proof}
	Please see Appendix B. 
\end{proof}
\begin{Lemma}
	There exists universal constants $C,K>0$ such that, for every $T\in\mathbb N$, the following inequality: 
	\begin{equation}
		\label{eq:stat}
		\begin{aligned}
			\mathbb E \{\|\hat{\mathbf R}(\mathbf{x}_i)-\mathbf R(\mathbf x_i)\| \} \!
			\le \!
			C K^2
			\big(
			\sqrt{\frac{N_r}{T}}
			+
			\frac{N_r}{T}
			\big)
			\|\mathbf R(\mathbf x_i)\| \!
			\triangleq \! \varepsilon(T),
		\end{aligned}
	\end{equation}
	holds, and in particular $\varepsilon(T)\to 0$ as $T\to\infty$.
\end{Lemma}
\begin{proof}
	Please refer to Appendix C.
\end{proof}
Combining Lemmas~2-3 with the triangle inequality yields the following bound on the deviation between the sample covariance matrix $\hat{\mathbf R}(\mathbf{x}_i)$ and the true covariance matrix ${\mathbf R}(\mathbf{x})$:
\begin{equation}
	\label{eq:uniform-short}
	\begin{aligned}
		&\mathbb E \{\|\mathbf R(\mathbf x)-\hat{\mathbf R}(\mathbf{x}_i)\| \}\\
				&\le
		\|\mathbf R(\mathbf x)-{\mathbf R}(\mathbf{x}_i)\|+	\mathbb E\{\|\mathbf R(\mathbf{x}_i)-\hat{\mathbf R}(\mathbf{x}_i)\|\}
		\\ &\le L_R \|\mathbf x-\mathbf{x}_i\| +\varepsilon(T).
	\end{aligned}
\end{equation}

{{
	To quantify the discrepancy between the surrogate objective
	$\hat g(\mathbf x)$ and the true objective $g(\mathbf x)$, we next
	establish a high-probability local error bound, conditional on the
	fixed quasi-static channel parameters and the current anchor
	$\mathbf x_i$.
		
	\begin{Lemma}
		\label{lem:surrogate_error_bound}
		Define the sample covariance
		estimation error at the current anchor $\mathbf x_i$ as
	$
			\boldsymbol\Delta_i
			\triangleq
			\hat{\mathbf R}(\mathbf x_i)
			-
			\mathbf R(\mathbf x_i).
	$
		Let $\varepsilon(T)$ be the covariance estimation bound given in
		Lemma~3. For any sufficiently large $T$ satisfying
		$0<\varepsilon(T)\le\sigma_n^2/4$, define the event
		\begin{equation}
			\label{eq:event_Ei}
			\mathcal E_i
			\triangleq
			\left\{
			\big\|
			\boldsymbol\Delta_i
			\big\|
			\le
			\sqrt{\sigma_n^2\varepsilon(T)}
			\right\}.
		\end{equation}
		Then,
		on the event $\mathcal E_i$, the following inequality holds
		simultaneously for all $\mathbf x\in\mathcal X$:
		\begin{equation}
			\label{eq:gap-global}
			\begin{aligned}
				\big|
				g(\mathbf x)-\hat g(\mathbf x)
				\big|
				&\le
				C_R
				\left(
				L_R\|\mathbf x-\mathbf x_i\|
				+
				\sqrt{\sigma_n^2\varepsilon(T)}
				\right)
				\triangleq
				\delta(\mathbf x),
			\end{aligned}
		\end{equation}
		where
$
			C_R
			\triangleq
			{
				2N_rL\|\boldsymbol\alpha\|^2
			}/{
				\sigma_n^4
			}.
$	Moreover, 		
\begin{equation}
		\label{eq:event_Ei_probability}
		\mathbb P(\mathcal E_i)
		\ge
		1-
		\sqrt{
			\frac{\varepsilon(T)}{\sigma_n^2}
		},
	\end{equation}
		where $\mathbb P(\mathcal E_i)\to1$ as
		$T\to\infty$.
	\end{Lemma}
	
	\begin{proof}
		Please see Appendix~\ref{app:surrogate_error_bound}.
	\end{proof}

	}}

	Lemma~\ref{lem:surrogate_error_bound} shows that, with probability at least
$1-\sqrt{\varepsilon(T)/\sigma_n^2}$, $g(\mathbf x)$ lies in a tube of width
$2\delta(\mathbf x)$ centered at $\hat g(\mathbf x)$, and the tube is
{narrowest near $\mathbf x_i$} because $\delta(\mathbf x)$ grows (approximately) linearly with
$\|\mathbf x-\mathbf x_i\|$. A toy example is depicted in Fig.~\ref{surragate} which visualizes this effect. 
Specifically, the light blue shaded region represents the interval $[\hat g(\mathbf x)-\delta(\mathbf x),\hat g(\mathbf x)+\delta(\mathbf x)]$.
Its width depends on $\|\mathbf x-\mathbf x_i\|$, i.e., when $\mathbf x$ stays near $\mathbf x_i$, the band is narrow and $\hat g(\mathbf x)$ closely tracks $g(\mathbf x)$; as $\mathbf x$ moves farther away, the band becomes wider, which indicates larger approximation error.
Consequently, when we optimize $\hat g(\mathbf x)$ using an
algorithm that keeps iterates within the neighborhood of $\mathbf x_i$, the predicted increase in the surrogate dominates the
tube width and thus may transfer to an actual increase of $g(\mathbf x)$. 
This observation explains why maximizing $\hat g(\mathbf x)$ is a principled proxy for improving the true objective $g(\mathbf x)$, which
is intuitive because, by Lemma~3 and Lemma~4, $\hat{\mathbf R}(\mathbf x_i)$ is an asymptotically consistent estimate of $\mathbf R(\mathbf x_i)$ (its estimation error vanishes as $T$ grows), and the gap between $\hat g(\mathbf x)$ and $g(\mathbf x)$ increases approximately linearly with $\|\mathbf x-\mathbf x_i\|$. 
Hence, as long as the iterates remain reasonably close to $\mathbf x_i$, the gains obtained by optimizing $\hat g(\mathbf x)$ would be very close to that by $g(\mathbf x)$.


{
In practice, under the two-timescale frame structure shown in Fig.~\ref{twotimescale}, the covariance \(\mathbf R(\mathbf x_i)\) is re-estimated at the current antenna-position anchor \(\mathbf x_i\) over each \(T\)-snapshot block. At the end of each block, the anchor is updated to \(\mathbf x_{i+1}\), which recenters the surrogate function  $\hat g(\mathbf x)$, and further tightens the local tube, thereby maintaining high surrogate fidelity along the trajectory.}

Having established the local fidelity of the local-anchor surrogate \( \hat{\mathbf R}(\mathbf x_i)\),
a natural question is whether one can further improve robustness by
aggregating covariance information from multiple anchors.
We address this question next by analyzing the historical-average
surrogate $ \hat{\mathbf R}_{\mathrm{avg}} \triangleq \frac{1}{M}\sum_{m=1}^{M}\hat{\mathbf R}(\mathbf x_m)$.

\subsection{Geometric Bias of Historical-Average Surrogates}

\label{subsec:hist_avg}

A seemingly natural extension of the local-anchor strategy is to
exploit all previously visited anchors and build a surrogate
objective from the arithmetic mean of their (blockwise) sample
covariances, i.e., from the historical average
$\hat{\mathbf R}_{\mathrm{avg}}$.  Intuitively,
such averaging may look attractive because it reduces the variance of
the sample covariance estimator and appears to provide a more “stable’’
model for optimization.

In practice, both the local-anchor surrogate \( \hat{\mathbf R}(\mathbf x_i)\) and the historical-average 
surrogate $\hat{\mathbf R}_{\mathrm{avg}}$ are implemented using the blockwise sample covariances
$\{\hat{\mathbf R}(\mathbf x_m)\}$ rather than the unknown population
covariances $\{\mathbf R(\mathbf x_m)\}$.  
However, Lemma~3 shows that, for each anchor, the sample covariance
$\hat{\mathbf R}(\mathbf x_m)$ deviates from $\mathbf R(\mathbf x_m)$ by a
statistical perturbation whose expected spectral norm vanishes as the
snapshot budget $T$ increases.  
This suggests that the structural difference between the two surrogates, i.e., $\hat{\mathbf R}(\mathbf x_i)$ and $\hat{\mathbf R}_{\mathrm{avg}}$,
primarily comes from how they combine information across different MA
geometries, rather than from the sampling noise itself.  
To make this geometric distinction explicit, we first compare the two
surrogates using true covariance matrices, and then assess
the implications for implementations based on sample covariances.

We now formalize the geometric comparison between the
historical-average model and the local-anchor model for the true covariance matrices.
Let $\mathbf x_\star\in\mathcal X$ denote the current antenna-position anchor and define\footnote{With a slight abuse of notation, we write $\mathbf R_0$ and
	$\mathbf R_\delta$ in place of $\mathbf R(\mathbf x_\star)$ and
	$\mathbf R(\mathbf x_\star+\boldsymbol\delta)$, respectively. Throughout
	the theorem, the anchor $\mathbf x_\star$ is fixed, so $\mathbf R_0$ can
	be viewed as a constant (local-anchor) covariance surrogate, whereas
	$\mathbf R_\delta$ denotes the true covariance evaluated at the displaced
	position $\mathbf x_\star+\boldsymbol\delta$. The dependence on
	$\mathbf x_\star$ and $\boldsymbol\delta$ can be made explicit as
	$\mathbf R(\mathbf x_\star)$ and $\mathbf R(\mathbf x_\star+\boldsymbol\delta)$
	whenever needed.}
	\begin{equation}
		\begin{aligned}
			\mathbf R_0
			&\triangleq \mathbf R(\mathbf x_\star), \quad \\
			\mathbf R_\delta
			&\triangleq
			\mathbf R(\mathbf x_\star+\boldsymbol\delta),
			\;
			\boldsymbol\delta\in\mathbb R^{N_r}.
		\end{aligned}
	\end{equation}
	Suppose that $M$ (not necessarily distinct)
	anchors $\{\mathbf x_m\}_{m=1}^{M}\subset\mathcal X$ have been selected and define the
	historical average covariance as
	\begin{equation}
		\begin{aligned}
			\widetilde{\mathbf R}
			&\triangleq
			\frac{1}{M}
			\sum_{m=1}^{M}\mathbf R(\mathbf x_m).
		\end{aligned}
	\end{equation}
	Then, we can obtain the following theorem.

{{\begin{Theorem}
			\label{thm:avg_vs_local}
			Assume that the historical-average covariance does not coincide
			with the local-anchor covariance at $\mathbf x_\star$, i.e.,
$
				\widetilde{\mathbf R}-\mathbf R_0
				\neq \mathbf 0.
$
			Then, there exist constants $c_0(\mathbf x_\star)>0$ and
			$\rho_0(\mathbf x_\star)>0$ such that, for all
			$\boldsymbol\delta\in\mathbb R^{N_r}$ satisfying
			$\|\boldsymbol\delta\|\le\rho_0(\mathbf x_\star)$, we have
			\begin{equation}
				\begin{aligned}
					\big\|
					\widetilde{\mathbf R}^{-1}
					-\mathbf R_\delta^{-1}
					\big\|
					&\ge
					\big\|
					\mathbf R_0^{-1}
					-\mathbf R_\delta^{-1}
					\big\|
					+c_0(\mathbf x_\star).
				\end{aligned}
			\end{equation}
		\end{Theorem}
		
		\begin{proof}
			Please see Appendix~\ref{app:avg_vs_local}.
\end{proof}}}

%

Theorem~\ref{thm:avg_vs_local} shows that, in a neighborhood of the current
anchor $\mathbf x_\star$, the inverse of the historical-average covariance
matrix, $\widetilde{\mathbf R}^{-1}$, is uniformly farther (measured by the spectral norm) from
the inverse covariance at any nearby position, i.e., 
$\mathbf R(\mathbf x_\star+\boldsymbol\delta)^{-1}$, than the local-anchor
inverse covariance $\mathbf R_0^{-1} = \mathbf R(\mathbf x_\star)^{-1}$.
More precisely, for all sufficiently small perturbations
$\boldsymbol\delta$, the distance
$\|\widetilde{\mathbf R}^{-1}-\mathbf R_\delta^{-1}\|$ exceeds
$\|\mathbf R_0^{-1}-\mathbf R_\delta^{-1}\|$ by at least a fixed margin
$c_0(\mathbf x_\star)>0$.  
This implies that when the MA array moves slightly away from the current
anchor, the local-anchor model $\mathbf R_0^{-1}$ always provides a better
local approximation to $\mathbf R_\delta^{-1}$ than the historical-average
model $\widetilde{\mathbf R}^{-1}$, which reveals an intrinsic geometric bias
of the historical-average surrogate around $\mathbf x_\star$.
Thus the local-anchor surrogate is naturally more compatible with
TR updates, whose effectiveness critically relies on an accurate
local model of $\mathbf R(\mathbf x)^{-1}$ along small position changes.

{{
		Although Theorem~\ref{thm:avg_vs_local} is formulated in terms of
		the population covariance matrices, the surrogates used in practice
		are constructed from blockwise sample covariances. For a fixed finite
		set of historical anchors,  we define
		\begin{equation}
			\label{eq:sample_surrogate_decomposition}
			\begin{aligned}
				\hat{\mathbf R}_{\mathrm{avg}}
				&\triangleq
				\frac{1}{M}
				\sum_{m=1}^{M}
				\hat{\mathbf R}(\mathbf x_m)=
				\widetilde{\mathbf R}
				+
				\widetilde{\boldsymbol\Delta},
			\end{aligned}
		\end{equation}
		where
$
			\widetilde{\boldsymbol\Delta}
			\triangleq
			\frac{1}{M}
			\sum_{m=1}^{M}
			\boldsymbol\Delta_m
$.
		Thus, the local-anchor covariance $\hat{\mathbf R}(\mathbf x_\star)$  contains the single-block
		statistical perturbation
				$\boldsymbol\Delta_\star
		\triangleq
		\hat{\mathbf R}(\mathbf x_\star)-\mathbf R_0
		$, whereas the
		historical-average covariance contains the averaged perturbation
		$\widetilde{\boldsymbol\Delta}$.
		By Lemma~3, let $\varepsilon_\star(T)$ and
		$\{\varepsilon_m(T)\}_{m=1}^{M}$ be nonnegative quantities satisfying
$
				\mathbb E
				\left[
				\big\|
				\boldsymbol\Delta_\star
				\big\|
				\right]
				\le
				\varepsilon_\star(T),
				\mathbb E
				\left[
				\big\|
				\boldsymbol\Delta_m
				\big\|
				\right]
				\le
				\varepsilon_m(T),
				m=1,\ldots,M,
$
		where
		$\varepsilon_\star(T)\to0$ and
		$\varepsilon_m(T)\to0$ as $T\to\infty$ for every fixed $m$.
		We further define
$
			\widetilde{\varepsilon}_M(T)
			\triangleq
			\frac{1}{M}
			\sum_{m=1}^{M}
			\varepsilon_m(T).
$
With the above notation, we next examine whether the comparison in
Theorem~\ref{thm:avg_vs_local} remains valid for
$\hat{\mathbf R}(\mathbf x_\star)$ and
$\hat{\mathbf R}_{\mathrm{avg}}$.

\begin{Theorem}
	\label{thm:sample_gap_stability}
	Assume that
	$\widetilde{\mathbf R}\neq\mathbf R_0$, and let
	$c_0(\mathbf x_\star)>0$ and
	$\rho_0(\mathbf x_\star)>0$ be the corresponding constants
	in Theorem~\ref{thm:avg_vs_local}. Then define
	\begin{equation}
		\label{eq:tau0_definition}
		\tau_0(\mathbf x_\star)
		\triangleq
		\min
		\left\{
		\frac{\sigma_n^2}{2},
		\frac{
			c_0(\mathbf x_\star)\sigma_n^4
		}{8}
		\right\}.
	\end{equation}
	Consider the small-perturbation event
	\begin{equation}
		\label{eq:sample_good_event}
		\mathcal E
		\triangleq
		\left\{
		\big\|
		\boldsymbol\Delta_\star
		\big\|
		\le
		\tau_0(\mathbf x_\star),
		\
		\big\|
		\widetilde{\boldsymbol\Delta}
		\big\|
		\le
		\tau_0(\mathbf x_\star)
		\right\}.
	\end{equation}
	Then, on the event $\mathcal E$, 
	the following inequality holds simultaneously for all
	$\boldsymbol\delta\in\mathbb R^{N_r}$ satisfying
	$\|\boldsymbol\delta\|\le\rho_0(\mathbf x_\star)$:
	\begin{equation}
		\label{eq:sample_strict_gap}
		\begin{aligned}
			\big\|
			\hat{\mathbf R}_{\mathrm{avg}}^{-1}
			-
			\mathbf R_\delta^{-1}
			\big\|
			&\ge
			\big\|
			\hat{\mathbf R}(\mathbf x_\star)^{-1}
			-
			\mathbf R_\delta^{-1}
			\big\|
			+
			\frac{c_0(\mathbf x_\star)}{2}.
		\end{aligned}
	\end{equation}
	The probability of the sufficient event $\mathcal E$ satisfies
	\begin{equation}
		\label{eq:sample_good_event_probability}
		\mathbb P(\mathcal E)
		\ge
		1
		-
		\frac{
			\varepsilon_\star(T)
			+
			\widetilde{\varepsilon}_M(T)
		}{
			\tau_0(\mathbf x_\star)
		}.
	\end{equation}
	Consequently, for every fixed finite $M$,
	$\mathbb P(\mathcal E)\to1$ as $T\to\infty$.
\end{Theorem}

\begin{proof}
	Please see Appendix~\ref{app:sample_gap_stability}.
\end{proof}

Theorem~\ref{thm:sample_gap_stability} shows that, throughout the
neighborhood
$\|\boldsymbol\delta\|\le\rho_0(\mathbf x_\star)$, the distance
between the historical-average inverse covariance
$\hat{\mathbf R}_{\mathrm{avg}}^{-1}$ and the true inverse covariance
$\mathbf R_\delta^{-1}$, i.e., $			\big\|
\hat{\mathbf R}_{\mathrm{avg}}^{-1}
-
\mathbf R_\delta^{-1}
\big\|$, is larger than the corresponding distance
between the local-anchor inverse covariance
$\hat{\mathbf R}(\mathbf x_\star)^{-1}$ and
$\mathbf R_\delta^{-1}$, i.e., $	\big\|
\hat{\mathbf R}(\mathbf x_\star)^{-1}
-
\mathbf R_\delta^{-1}
\big\|$, by at least
$c_0(\mathbf x_\star)/2$, with probability lower bounded by
\eqref{eq:sample_good_event_probability}. This probability approaches
one as $T\to\infty$.

To relate this geometric result to the objective
functions $\hat g(\mathbf x)$ evaluated in practice, we consider a candidate position
$\mathbf x=\mathbf x_\star+\boldsymbol\delta$ satisfying
$\|\boldsymbol\delta\|\le\rho_0(\mathbf x_\star)$, and the sample-based local-anchor and historical-average
surrogate objectives are defined as
$
		\hat g_{\mathrm{loc}}(\mathbf x)
		\triangleq
		\mathbf h_0(\mathbf x)^H
		\hat{\mathbf R}(\mathbf x_\star)^{-1}
		\mathbf h_0(\mathbf x),
		\hat g_{\mathrm{avg}}(\mathbf x)
		\triangleq
		\mathbf h_0(\mathbf x)^H
		\hat{\mathbf R}_{\mathrm{avg}}^{-1}
		\mathbf h_0(\mathbf x),
$
respectively.
Since the true objective at
$\mathbf x=\mathbf x_\star+\boldsymbol\delta$ is
$g(\mathbf x)
=\mathbf h_0(\mathbf x)^H
\mathbf R_\delta^{-1}
\mathbf h_0(\mathbf x)$, the two approximation errors satisfy
\begin{equation}
	\label{eq:sample_based_objective_error_bounds1}
	\begin{aligned}
		\big|
		\hat g(\mathbf x)
		-
		g(\mathbf x)
		\big|
		&\le
		\big\|
		\hat{\mathbf R}(\mathbf x_\star)^{-1}
		-
		\mathbf R_\delta^{-1}
		\big\|
		\big\|
		\mathbf h_0(\mathbf x)
		\big\|^2,
	\end{aligned}
\end{equation}
\begin{equation}
	\label{eq:sample_based_objective_error_bounds2}
	\begin{aligned}
		\big|
		\hat g_{\mathrm{avg}}(\mathbf x)
		-
		g(\mathbf x)
		\big|
		&\le
		\big\|
		\hat{\mathbf R}_{\mathrm{avg}}^{-1}
		-
		\mathbf R_\delta^{-1}
		\big\|
		\big\|
		\mathbf h_0(\mathbf x)
		\big\|^2.
	\end{aligned}
\end{equation}
Therefore, with probability lower bounded by
\eqref{eq:sample_good_event_probability}, the error bound for
$\hat g_{\mathrm{avg}}(\mathbf x)$ exceeds that for
$\hat g(\mathbf x)$ by at least
$c_0(\mathbf x_\star)\|\mathbf h_0(\mathbf x)\|^2/2$.
The uniformly smaller upper error bound of
$\hat g(\mathbf x)$ provides a theoretical
justification for using the current-anchor sample covariance
$\hat{\mathbf R}_{\mathrm{loc}}$ in local TR updates. This conclusion
is also consistent with the numerical results reported in
Section~V.

}}

\section{Proposed Algorithm}

Our goal in this section is to design an efficient algorithm to solve problem \eqref{P1}. The main challenges are (i) it is highly nonconvex 
as the optimization variable $\mathbf x$ appears in the exponential terms of 
each channel coefficient, 
and (ii) although the mismatch between the surrogate objective $\hat g(\mathbf x)$ and the true objective $ g(\mathbf x)$ is bounded, it is still non-negligible for a practical value of $T$. 
These observations motivate us to use a TR strategy that explicitly controls the step size so that each update remains within a neighborhood of $\mathbf x_i$ where the surrogate is reliable. 
We therefore propose a low-complexity iterative PTRSO algorithm, which at each iteration maximizes a quadratic surrogate function within a TR, then projects the trial point onto the feasible set $\mathcal X$, and applies a ratio test to accept the step and adjust the TR radius until a  stationary solution is found.

\subsection{PTRSO Algorithm}
We first outline the workflow of the TR strategy used by PTRSO and show how the structure of $\hat g(\mathbf x)$ yields a low-complexity implementation.

Since the surrogate \(\hat g(\mathbf x)\) is smooth, we adopt a local quadratic approximation to capture both its slope and curvature information while keeping the subproblem tractable. Specifically, for a trial step \(\mathbf x_p\) of the current $k$-th iteration, we use a local quadratic model by Taylor expansion of \(\hat g(\mathbf x)\) at \(\mathbf x^k\), i.e.,
\begin{equation}
	\label{eq:qk}
	q_k(\mathbf x_p)= \hat g(\mathbf x^k) + \mathbf g_k^{\!T}\mathbf x_p+\tfrac12\,\mathbf x_p^{T}\mathbf H_k\mathbf x_p,
\end{equation}
where $\mathbf g_k=\nabla \hat g(\mathbf x^k)$ and $ \mathbf H_k=\nabla^2 \hat g(\mathbf x^k)$. 
Notably, the gradient and Hessian both admit closed-form expressions, which are given by 
\begin{equation}
	\label{eq:grad}
	\nabla \hat g(\mathbf x)
	= 2\Re\left\{\mathbf J(\mathbf x)^{\!H}\,
	\hat{\mathbf R}(\mathbf x_i)^{-1}\,\mathbf h_0(\mathbf x)\right\},
\end{equation}
and
\begin{equation}
	\label{eq:hess}
	\begin{aligned}
		\nabla^2 &  \hat g(\mathbf x)
		= 2\Re\Big\{
		\mathbf J(\mathbf x)^{\!H}\hat{\mathbf R}(\mathbf x_i)^{-1}\mathbf J(\mathbf x) \\
		&
		-\,k^2\,\operatorname{diag}\!\big(\operatorname{diag}(c(x_1),\ldots,c(x_{N_r}))^{\!H}
		\,\hat{\mathbf R}(\mathbf x_i)^{-1}\mathbf h_0(\mathbf x)\big)
		\Big\},
	\end{aligned}
\end{equation}
respectively, where $\mathbf J(\mathbf x) \! = \! \frac{\partial \mathbf h_0(\mathbf x)}{\partial \mathbf x}
=\operatorname{diag}\{b(x_1),\ldots,b(x_{N_r})\}$ with
$b(x_n)=\sum_{\ell=1}^{L} j k \alpha_\ell \sin\theta_\ell\,e^{j k x_n \sin\theta_\ell}$, 
and $c(x_n)=\sum_{\ell=1}^{L}\alpha_\ell\sin^2\theta_\ell\,e^{j k x_n \sin\theta_\ell}$.
As can be observed, both (\ref{eq:grad}) and (\ref{eq:hess}) share a simple structure, i.e.,  a diagonal term multiplied by $\hat{\mathbf R}^{-1}(\mathbf x_i)$ and another diagonal term, minus a diagonal correction. 
Leveraging this structure, the proposed algorithm exploits Hessian-vector products to avoid explicit Hessian construction and factorization. 
	In particular, the surrogate Hessian takes the form $\mathrm{diag}(\cdot)\hat{\mathbf R}(\mathbf x_i)^{-1}\mathrm{diag}(\cdot)$ minus a diagonal term, so each Hessian-vector product reduces to one multiplication by $\hat{\mathbf R}(\mathbf x_i)^{-1}$ followed by elementwise scalings. 
	As a result, the per-iteration complexity scales as $\mathcal O(N_r^2)$, instead of the $\mathcal O(N_r^3)$ cost caused by forming and decomposing a dense Hessian. This makes the proposed method well suited for large arrays, and the detailed complexity analysis is provided in Section~V. 

Then, based on the above quadratic model (\ref{eq:qk}) and using a given radius $\Delta_k>0$, we propose to solve the following TR subproblem in each iteration of the proposed PTRSO algorithm: 
\begin{equation}
	\label{eq:TR-sub}
	\max_{\|\mathbf x_p\|\le\Delta_k}\; q_k(\mathbf x_p).
\end{equation}
The above subproblem can be efficiently handled by the Steihaug-conjugate-gradient (CG) method \cite{book1}, which only requires Hessian-vector products and thus naturally fits the above structure. 
Meanwhile, the TR radius $\Delta_k$ limits the step length and maintains the iterate within a neighborhood where the surrogate model remains reliable.
After computing a candidate step $\mathbf x_p^k$,
we form the next trial point by $\mathbf x_{\mathrm{trial}}=\mathbf x^k+\mathbf x_p^k$ and ensures the feasibility of problem (\ref{P1}) by further applying a simple Euclidean projection step onto the linear inequality set $\mathcal X$, i.e., 
\begin{equation}
	\label{PAVA1}
\mathbf x^{\mathrm{proj}} \triangleq\operatorname{Proj}_{\mathcal X}(\mathbf x_{\mathrm{trial}})=\arg\min_{\mathbf x \in \mathcal X}\ \|\mathbf x-\mathbf x_{\mathrm{trial}}\|^2.
\end{equation}
Note that this projection step admits a linear complexity implementation. 
In particular, problem~\eqref{PAVA1} 
is equivalent to the following bounded isotonic regression problem under the variable change of $y_n=x_n-n d$:
\begin{equation}
	\label{PAVA}
	\begin{aligned}
\min_{\mathbf y} &\quad \sum_{n=1}^{N_r}\big(y_n-(x_{\mathrm{trial},n}-n d)\big)^2 \\
\text{s.t.}&\quad 
y_1\le y_2\le\cdots\le y_{N_r},\ \ -d\le y_i\le D_x-N_r d,\;\forall i,
\end{aligned}
\end{equation}
and it can be solved by the classical pool adjacent violators algorithm (PAVA) in $\mathcal O(N_r)$ time \cite{PAVA}, after which the solution can be recovered by $x_n^{\mathrm{proj}}=y_n+n d, \forall n$.

After obtaining the Steihaug-CG step $\mathbf x_p^k$, we first form
the projected TR step
\begin{equation}
	\label{eq:projected_tr_step}
	\mathbf d_{\mathrm{TR}}^k
	\triangleq
	\operatorname{Proj}_{\mathcal X}
	\big(
	\mathbf x^k+\mathbf x_p^k
	\big)
	-\mathbf x^k.
\end{equation}
Since the projection may reduce the model increase provided by
$\mathbf x_p^k$,  we additionally construct the feasible Cauchy step
$
	\mathbf d_{\mathrm C}^k
	\triangleq
	\operatorname{Proj}_{\mathcal X}
	\big(
	\mathbf x^k+\alpha_k\mathbf g_k
	\big)
	-\mathbf x^k,
$ where the step size
$0<\alpha_k\le\Delta_k/\|\mathbf g_k\|$
is selected by backtracking such that
\begin{equation}
	\label{eq:cauchy_model_increase}
	q_k(\mathbf d_{\mathrm C}^k)-q_k(\mathbf0)
	\ge
	\frac{\|\mathbf d_{\mathrm C}^k\|^2}{2\alpha_k}.
\end{equation}
The condition
$\alpha_k\le\Delta_k/\|\mathbf g_k\|$
keeps $\mathbf d_{\mathrm C}^k$ within the current TR $\Delta_k$, and
\eqref{eq:cauchy_model_increase} ensures a positive model increase
at every nonstationary iterate.
The trial step is then selected as
\begin{equation}
	\label{eq:trial_step_selection}
	\mathbf d^k
	\triangleq
	\begin{cases}
		\mathbf d_{\mathrm{TR}}^k,
		&
		q_k(\mathbf d_{\mathrm{TR}}^k)-q_k(\mathbf0)
		\ge
		\kappa_0
		\big[
		q_k(\mathbf d_{\mathrm C}^k)-q_k(\mathbf0)
		\big], 
		\\
		\mathbf d_{\mathrm C}^k,
		&
		\text{otherwise},
	\end{cases}
\end{equation}
where $\kappa_0\in(0,1]$. 

Next, to quantify how well the quadratic model (\ref{eq:qk}) predicts the actual surrogate progress, we compare its predicted model increase with the realized surrogate improvement. Specifically, we define the model predicted gain as
\begin{equation}
	\label{pred}
	\mathrm{pred}_k=q_k(\mathbf d^k)-q_k(\mathbf 0),
\end{equation}
where $\mathbf d^k=\mathbf x^{\mathrm{proj}}-\mathbf x^k $ is the current step,
while the realized surrogate gain is 
\begin{equation}
	\label{ared}
	\mathrm{ared}_k=\hat g(\mathbf x^k+\mathbf d^k)-\hat g(\mathbf x^k).
\end{equation}
We then define
\begin{equation}
	\label{rhok}
	\rho_k =\frac{\mathrm{ared}_k}{\mathrm{pred}_k}, 
\end{equation}
which compares the actual surrogate improvement to the model‐predicted improvement. It is an important parameter that can be used to decide step acceptance and adjust the TR radius $\Delta_k$. 
When \(\rho_k\) is close to one, the model \(q_k\) provides an accurate local approximation, in which case the current step $\mathbf d^k$ is accepted and the radius \(\Delta_k\) can be enlarged.
However, when  \(\rho_k\) is very small or even negative, the predicted gain is not realized, which indicates that the model has been trusted too far and the current step $\mathbf d^k$ is rejected, in this case \(\Delta_k\) should be reduced.
This feedback mechanism dynamically adjusts \(\Delta_k\) so that the TR stays at a scale where the quadratic model remains reliable.
The overall procedure for solving problem \eqref{P1} is summarized in Algorithm~1, which is guaranteed to converge to the set of stationary solutions of problem (\ref{P1})\cite{book1}.

\subsection{Intuitive Explanation of the Proposed PTRSO Algorithm}
\begin{figure}[t]
	\centering
	\includegraphics[width=2.4in]{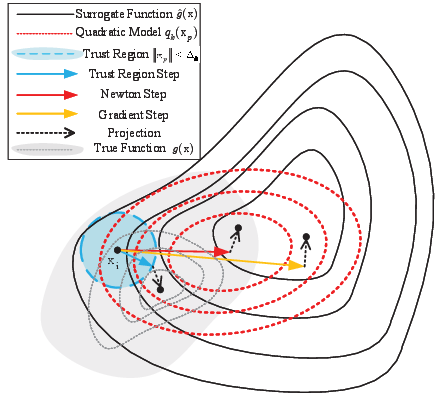}
	\caption{Toy example of one step iteration geometry.}
	\label{fig:PTRSO-geometry}
	\vspace{-0mm}
\end{figure}

%

	To provide further insights, we compare the proposed PTRSO algorithm with the projected gradient descent (PGD) and the Newton method with projection-based line search. In both baseline methods, a search direction is determined first, followed by a line search over the scalar step size along that direction, and the resulting point is projected onto the feasible set $\mathcal{X}$.
	However, line searches can accept long steps, and the repeated cycle of “search-then-project” often produces zigzagging moves near the feasible boundary. Both effects tend to carry the iterate far away from the anchor \(\mathbf x_i\). 
	As discussed in Section II, with fixed snapshot budget \(T\), the surrogate error \(|g(\mathbf x)-\hat g(\mathbf x)|\) satisfies the local estimate in \eqref{eq:gap-global}, which grows approximately linearly with \(\|\mathbf x-\mathbf x_i\|\). 
	Consequently, both baseline methods may drift away from the anchor $\mathbf x_i$, enlarge the surrogate error, and thereby weaken the improvement in the true objective, as illustrated by the toy example in Fig.~\ref{fig:PTRSO-geometry}.

	In contrast, the proposed PTRSO algorithm solves the TR subproblem~\eqref{eq:TR-sub} under the radius constraint $\|\mathbf{x}_p\|\le\Delta_k$. This confines each update within a neighborhood where the quadratic model~\eqref{eq:qk} remains reliable and the mismatch between the surrogate $\hat g(\mathbf x)$ and true objective $g(\mathbf x)$ is small, as depicted in Fig.~\ref{fig:PTRSO-geometry}. The trial point $\mathbf{x}_{\mathrm{trial}}=\mathbf{x}^k+\mathbf{x}_p$ is then projected back into $\mathcal{X}$, which avoids the repeated project-search zigzag that would otherwise enlarge $\|\mathbf{x}-\mathbf{x}_i\|$. Also, the step is accepted or rejected based on the gain ratio $\rho_k$ in~\eqref{rhok}, which adaptively regulates the TR radius and prevents excessively long moves that could deteriorate the local error control. 
	After each block of $T$ snapshots, the sample covariance is refreshed at the new anchor according to the two-timescale framework, i.e., from $\hat{\mathbf{R}}(\mathbf{x}_i)$ to $\hat{\mathbf{R}}(\mathbf{x}_{i+1})$, ensuring that the surrogate remains centered around the latest design point and retains its fidelity along the optimization trajectory.

\begin{algorithm}[t]
	\caption{Proposed PTRSO Algorithm}
	\label{alg:PTRSO}
	\renewcommand{\algorithmicrequire}{\textbf{Input:}}
	\renewcommand{\algorithmicensure}{\textbf{Output:}}
	\begin{algorithmic}[1]
		\REQUIRE Feasible set $\mathcal X=\{\mathbf x:\mathbf U\mathbf x\preceq\mathbf l\}$; initial point $\mathbf x^0\!\in\!\mathcal X$; initial radius $\Delta_0$ and maximum radius $\Delta_{\max}$; thresholds $0\le \eta<\eta_1<\eta_2<1$; factors $\gamma_{1}\!\in\!(0,1)$ and $\gamma_{2}\!>\!1$; tolerances $\epsilon_1$; maximum iterations $K_{\max}$.
		\ENSURE Optimized APV $\mathbf x_\star$.
		\STATE Set $k\!\leftarrow\!0$.
		\WHILE{$k<K_{\max}$}
		\STATE Compute
		$\{\hat g(\mathbf x^k),\mathbf g_k,\mathbf H_k\}$ and obtain the quadratic model $q_k(\mathbf x_p)$ using \eqref{eq:qk}--\eqref{eq:hess}. 
		\STATE Obtain $\mathbf x_p^k$ by solving \eqref{eq:TR-sub} via the Steihaug-CG method.
		\STATE Obtain $\mathbf x_{\text{trial}}=\mathbf x^k+\mathbf x_p^k$, compute $\mathbf x^{\text{proj}}$ via PAVA by solving (\ref{PAVA}), and select $\mathbf d^k$ through \eqref{eq:trial_step_selection}.
			\IF{$\|\mathbf d_{\mathrm C}^k\|/\alpha_k\le\epsilon_1$}
		\STATE \textbf{break}.
		\ENDIF
		\STATE Obtain $\mathrm{pred}_k$, $\mathrm{ared}_k$ and $\rho_k$ via \eqref{pred}--\eqref{rhok}.
		\STATE 
		$\small
		\Delta_{k+1}=\begin{cases}
			\gamma_{1}\,\Delta_k, & \rho_k\le \eta_1,\\
			\min\{\gamma_{2}\Delta_k,\Delta_{\max}\}, & \rho_k\ge \eta_2 \ \text{and}\ \|\mathbf d^k\|=\Delta_k,\\
			\Delta_k, & \text{otherwise}.
		\end{cases}
		$
		\STATE 		$
		\mathbf x^{k+1}=
		\begin{cases}
			\mathbf x^k+\mathbf d^k, & \rho_k>\eta,\\
			\mathbf x^k, & \text{otherwise}.
		\end{cases}
		$
		\STATE $k\!\leftarrow\!k+1$.
		\ENDWHILE
		\RETURN $\mathbf x_\star=\mathbf x^k$.
	\end{algorithmic}
\end{algorithm}

\begin{figure}[t]
	\centering
	\includegraphics[width=3in]{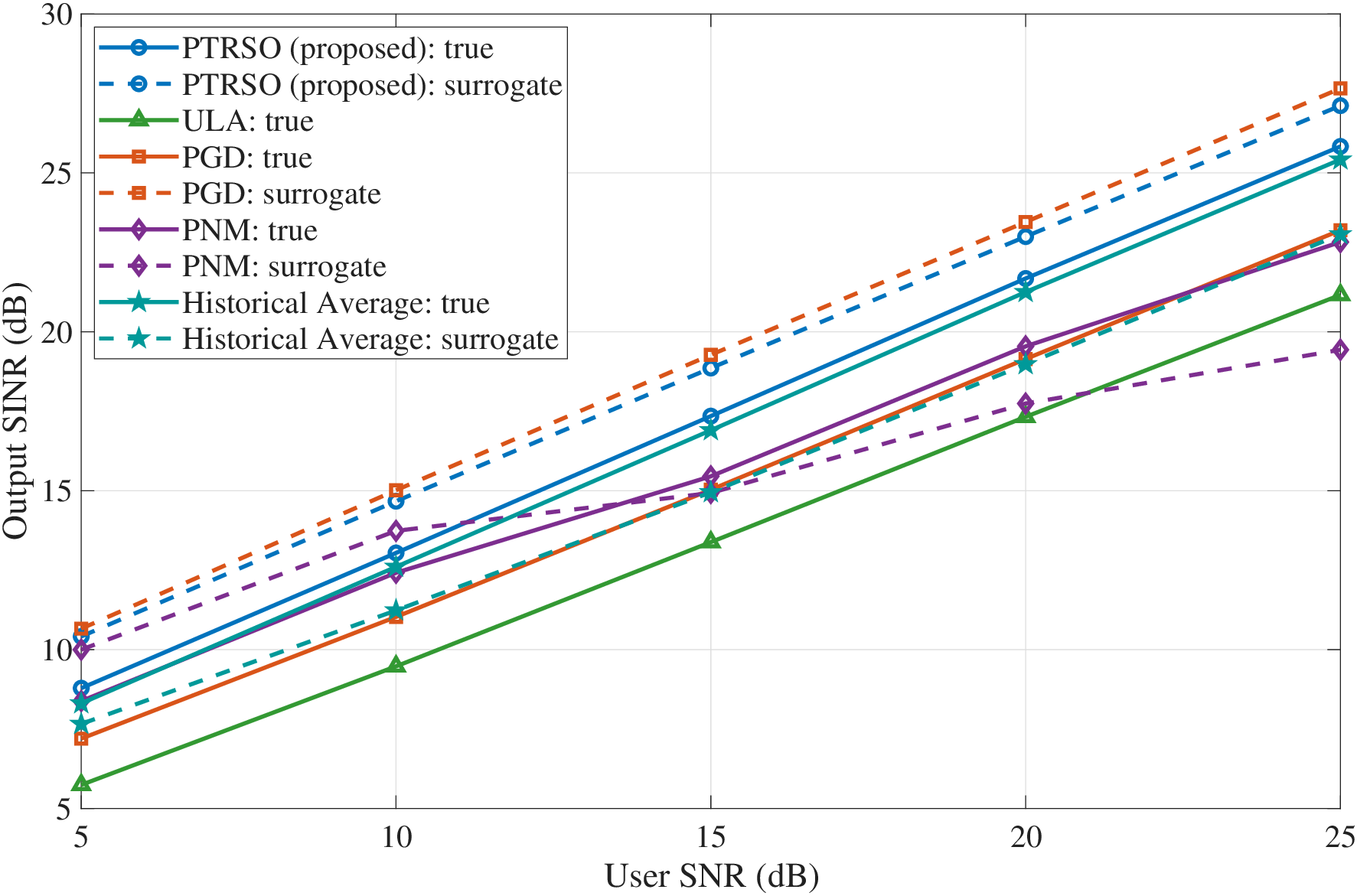}\vspace{-0.0cm}
	\caption{Output SINR (true and surrogate objective) versus user SNR.}
	\label{fig:sinr_snr}
\end{figure}

\begin{figure}[t]
	\centering
	\includegraphics[width=3in]{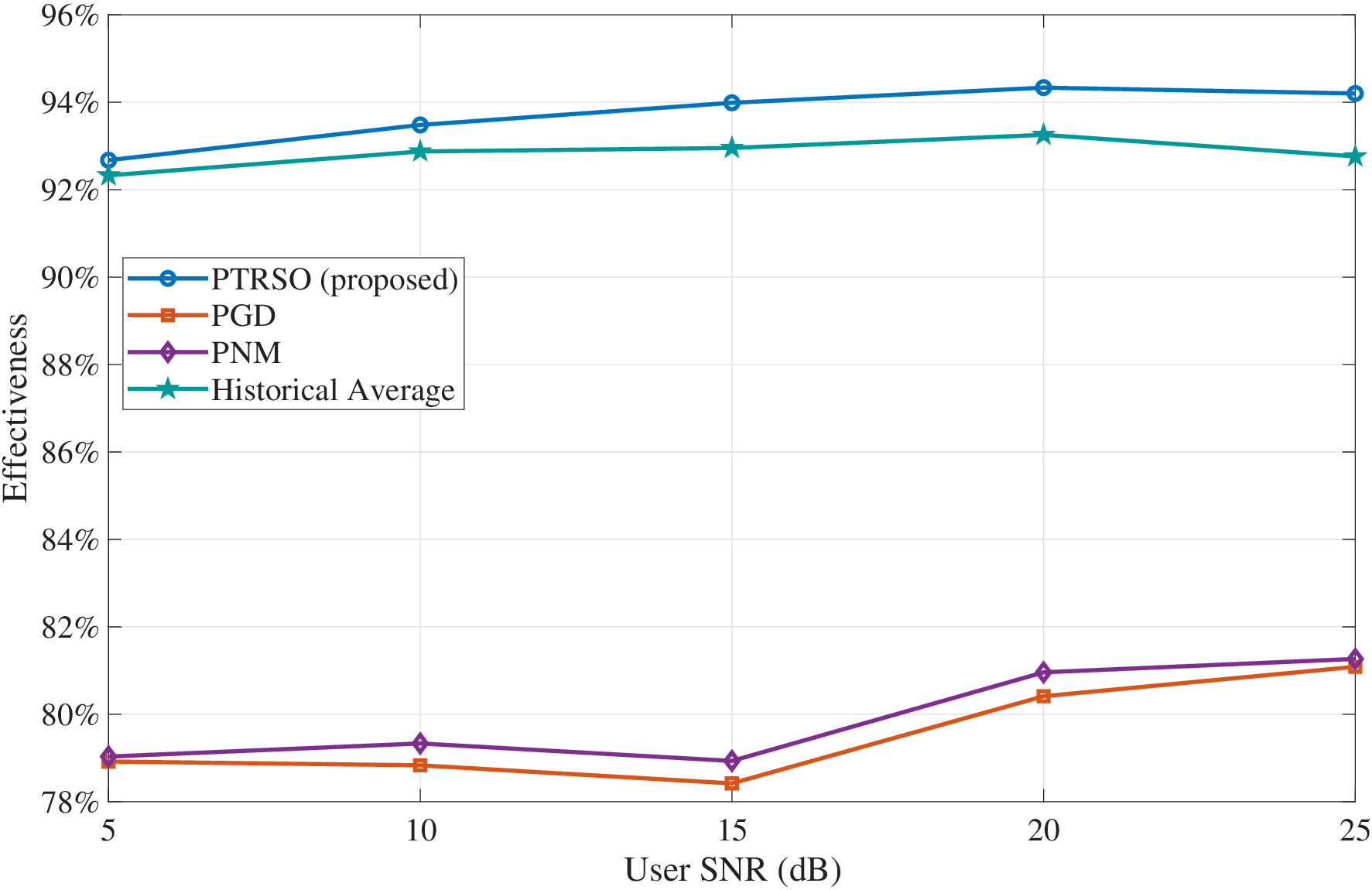}\vspace{-0.0cm}
	\caption{Effectiveness versus user SNR.}
	\label{fig:effectiveness}
\end{figure}

\begin{figure}[t]
	\centering
	\includegraphics[width=3in]{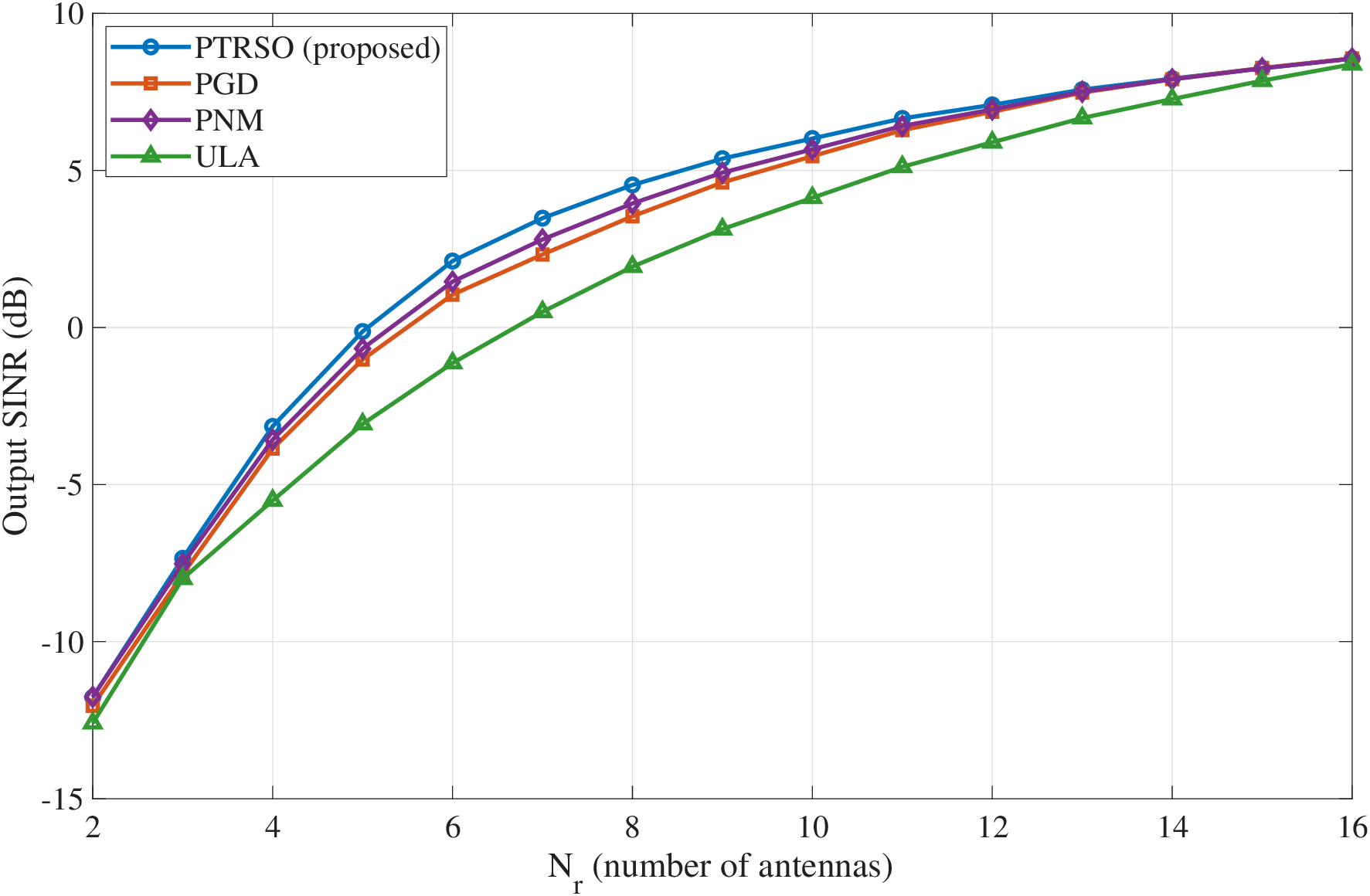}\vspace{-0.0cm}
	\caption{Output SINR versus number of receive antennas.}
	\label{fig:sinr_Nr}
\end{figure}

\begin{figure}[t]
	\centering
	\includegraphics[width=3in]{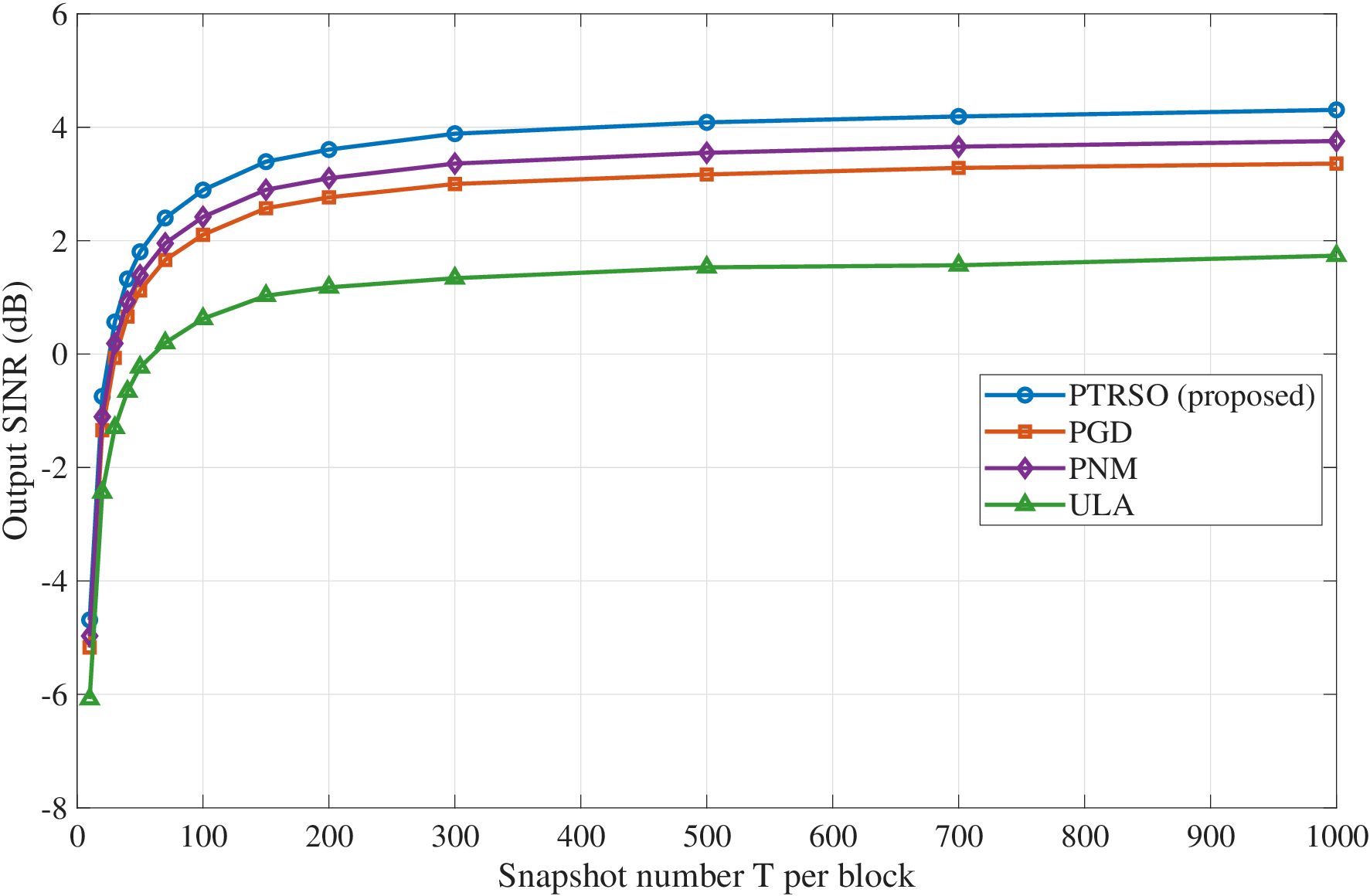}\vspace{-0.0cm}
	\caption{Output SINR versus number of snapshots per block.}
	\label{fig:sinr_T}
\end{figure}

\begin{figure}[t]
	\centering
	\includegraphics[width=3in]{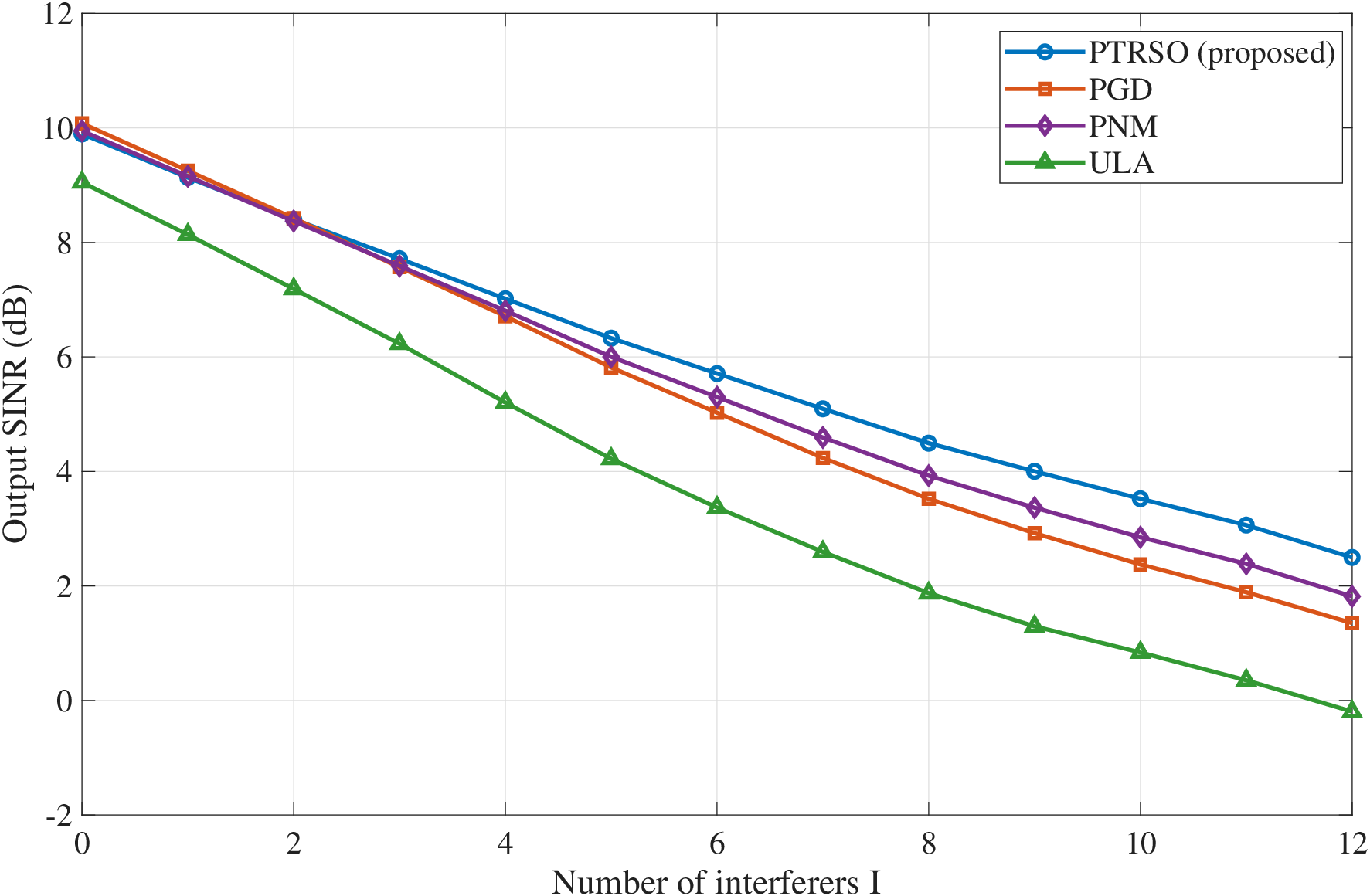}\vspace{-0.0cm}
	\caption{Output SINR versus number of jammers.}
	\label{fig:sinr_I}
\end{figure}

\begin{figure}[t]
	\centering
	\includegraphics[width=3in]{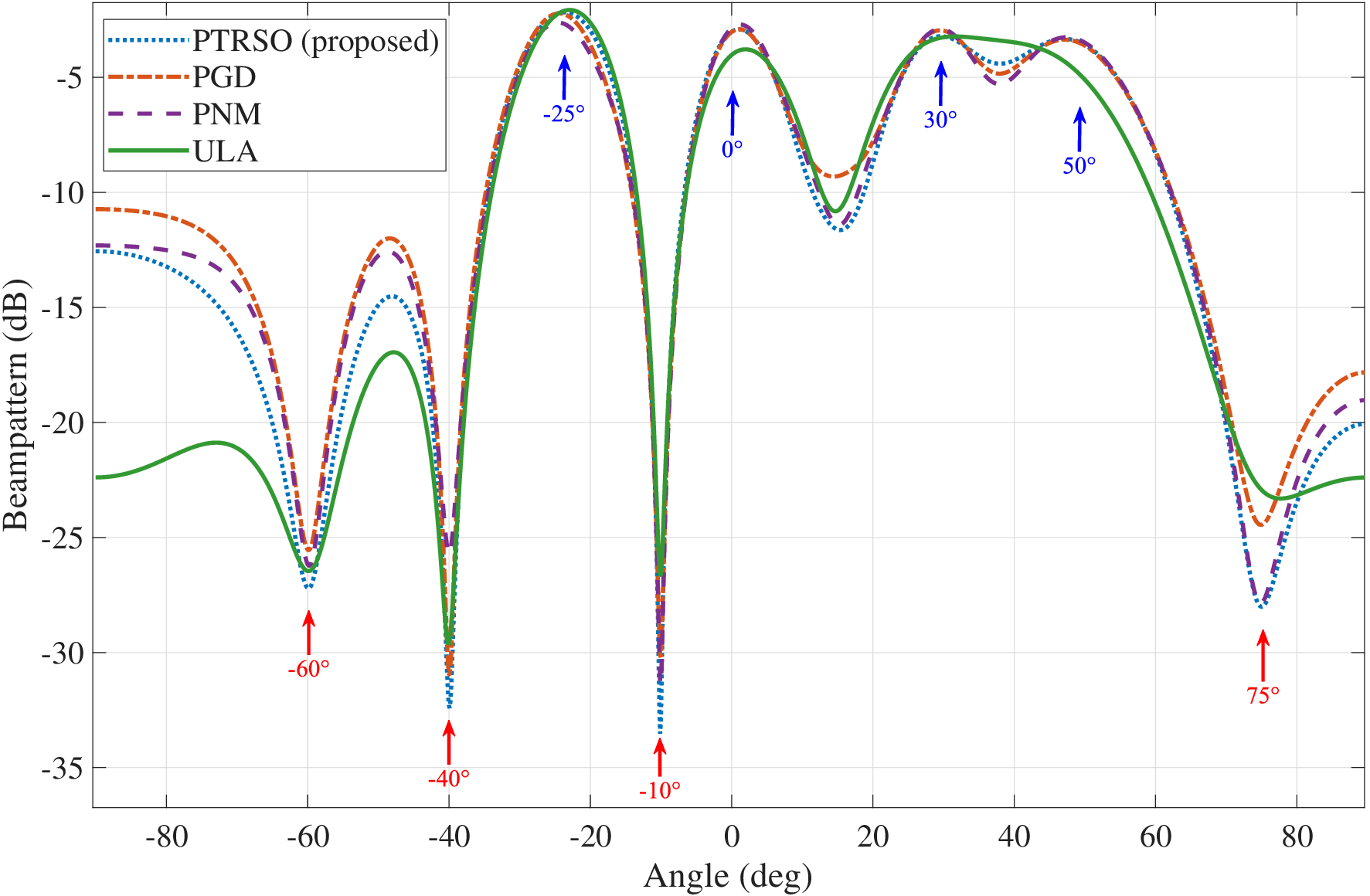}\vspace{-0.0cm}
	\caption{Comparison of beampatterns with different antennas’ positions.}
	\label{fig:beampattern}
\end{figure}

\vspace{-0.0cm}
\section{Numerical Results}
This section presents numerical results to validate the effectiveness of the proposed MA-enabled anti-jamming design under the MVDR framework.
We consider a receive array with \(N_r=8\) MAs, and
the antenna aperture is set to \(D_x=8\lambda\).
The initial anchor is the ULA point \(\mathbf x_0=[0,d,\ldots,(N_r\!-\!1)d]^T\).
The desired channel comprises \(L=8\) paths with azimuth angles \(\theta_\ell\sim\mathrm{Unif}[-\pi/2,\pi/2]\) and channel gains \(\alpha_\ell\sim\mathcal{CN}(0,1/(2L))\).
Besides, \(I=8\) jammers/interferers are considered in this simulation with AoAs \(\phi_i\sim\mathrm{Unif}[-\pi/2,\pi/2]\) and channel gains \(\zeta_i\sim\mathcal{CN}(0,1)\).
Other parameters are set as follows unless
otherwise specified:
\(\operatorname{SNR} = \sigma_s^2/\sigma_n^2\), \(\sigma_i^2=10 \sigma_s^2\), \(\eta_1=0.25\), \(\eta_2=0.75\), \(\gamma_{1}=0.25\), \(\gamma_{2}=2\), \(\epsilon_1=10^{-6}\), and the sample covariance \(\hat{\mathbf R}(\mathbf x_i)\) is formed from \(T=100\) snapshots at \(\mathbf x_i\).
We report both the true output SINR, computed using $\mathbf{R}(\mathbf{x})$, and the surrogate SINR, computed using $\hat{\mathbf{R}}(\mathbf{x}_i)$. 
Each point in the results is averaged over \(1.5\times 10^{4}\) independent Monte Carlo trials.
The following baselines are included:
1) PGD with Armijo backtracking;
	2) Newton method with projection-based line search, denoted as PNM;
	3) unoptimized FPA at \(\mathbf x_0\) with half-wavelength spacing (i.e., ULA);
4) Proposed PTRSO algorithm, but using the historical average surrogate $\hat{\mathbf R}_{\mathrm{avg}}$.

First, we show in Fig.~\ref{fig:sinr_snr} the output SINR versus the user SNR
for the considered algorithms, where both the true SINR and its surrogate
counterpart are plotted. 
Across the entire SNR range, the proposed PTRSO algorithm achieves the
highest true output SINR, providing about \(3\)~dB SINR gain over PGD/PNM and
around \(4\)-\(5\)~dB gain over the ULA baseline.
We also observe that PGD often attains a larger surrogate objective than the
proposed method, and both PNM and the historical-average scheme exhibit
similarly optimistic surrogate SINR values, yet all of them achieve lower
true SINR than PTRSO, which is consistent with our analysis in Section IV.
Fig.~\ref{fig:effectiveness} illustrates  the {effectiveness} of the considered algorithms, defined as the percentage of Monte Carlo trials (out of \(1.5\times10^4\) per SNR point) in which the optimized MA configuration yields a higher true SINR than the ULA baseline. 
As shown, the proposed PTRSO algorithm attains the highest effectiveness, over $92\%$ for all SNR values, followed by PNM and PGD. 
The historical-average scheme exhibits lower effectiveness than the proposed local-anchor surrogate, and
this result is consistent with the analysis of the geometric bias characterized in Section~III.
Overall, these trends confirm that explicit step-size control via the
TR strategy, combined with the local-anchor covariance surrogate,
significantly improves the transfer from surrogate gains to true performance
gains under finite-snapshot covariance estimation.

Next, Fig.~\ref{fig:sinr_Nr} shows the output SINR versus the number of receive
antennas $N_r$, with the total aperture fixed at $D_x=8 \lambda$ and the user SNR set as 0 dB. 
As shown, the proposed PTRSO algorithm consistently attains the highest SINR across
the whole range of $N_r$. The PGD and PNM methods yield comparable but lower performance, while the ULA baseline performs the worst. 
The performance gap is most pronounced for moderate array sizes, where the
MA geometry still provides substantial additional spatial DoFs
compared with the ULA. 
When $N_r$ is very small (e.g., $N_r=2$), all schemes are fundamentally
limited by the lack of spatial DoFs, and the gains are mostly contributed by the MVDR beamforming. When $N_r$
is very large (e.g., $N_r=16$), the antennas become densely packed within
the fixed aperture and the SINR gain of adding more elements
diminishes. 
Fig.~\ref{fig:sinr_T} plots the output SINR as a function of the snapshot number $T$
per block. 
It is observed that all schemes exhibit a monotonic SINR improvement with $T$ that
gradually saturates, while the proposed PTRSO algorithm maintains a clear
advantage over PGD, PNM, and the ULA baseline for all snapshot budgets, which
agrees with Lemma~3.
Fig.~\ref{fig:sinr_I} illustrates the output SINR versus the number
of jammers $I$. 
As expected, the SINR of all schemes degrades as more interferers are present,
but the proposed PTRSO algorithm consistently yields the highest SINR over
the entire range of $I$, followed by PNM and PGD, with the ULA baseline
remaining the worst.

To gain further insight, Fig.~\ref{fig:beampattern} compares the beampatterns of the considered schemes. 
For illustration, we consider a scenario where the desired user channel
contains \(L=4\) paths with azimuth angles
\(\theta_\ell \in \{-25^\circ,0^\circ,30^\circ,50^\circ\}\), which are
marked by blue arrows in the figure. 
In addition, \(I=4\) jammers/interferers are present with AoAs
\(\phi_i \in \{-60^\circ,-40^\circ,-10^\circ,75^\circ\}\), indicated by red
arrows.
We observe that all MA-based schemes form strong
mainlobes around the desired directions, while placing deep spatial nulls
close to the interferer angles. 
The proposed PTRSO solution achieves the deepest nulls at \(\{-60^\circ,-40^\circ,-10^\circ,75^\circ\}\) and slightly lower sidelobe levels in the vicinity of the
desired paths, which explains its superior output SINR. 
By contrast, the fixed ULA exhibits noticeably shallower notches at the
jammers' AoAs, thereby leading to higher residual interference power.

Finally, we evaluate the computational complexity of each algorithm. 
	For the proposed PTRSO, the dominant cost arises from solving the TR subproblem using the Steihaug-CG method, 
	which requires $\kappa$ Hessian-vector products, each with a cost of $\mathcal O(N_r^2)$. 
	Evaluating $\{\hat g,\nabla\hat g\}$ incurs an additional cost of $\mathcal O(N_r^2+N_r L)$, 
	while the complexity of the projection step $\operatorname{Proj}_{\mathcal X}(\cdot)$ is $\mathcal O(N_r)$ and thus negligible. 
	Consequently, the per-iteration complexity is $\mathcal O(\kappa N_r^2 + N_r L)$, 
	and the total complexity over $K$ iterations is $\mathcal O\!\big(K(\kappa N_r^2 + N_r L)\big)$. 
	For PGD with Armijo backtracking, the main cost stems from the line search, 
	which on average performs $\bar b$ evaluations of $\{\hat g,\nabla\hat g\}$, 
	each of cost $\mathcal O(N_r^2+N_r L)$. 
	Hence, the total complexity is $\mathcal O\!\big(K\,\bar b (N_r^2+N_r L)\big)$. 
	For the PNM baseline, the Hessian is explicitly formed and factorized via Cholesky decomposition at each iteration, 
	resulting in a per-iteration cost of $\mathcal O(N_r^3)$ and thus the total cost is $\mathcal O(K N_r^3)$. 
	Overall, PTRSO attains a per-iteration complexity comparable to that of PGD while avoiding the cubic scaling of PNM, thereby achieving a favorable balance between computational efficiency and optimization accuracy.

\section{Conclusion}
This paper investigated the challenge of MA-enabled anti-jamming reception under the practical constraint of unknown jamming channels. We introduced a snapshot-based covariance surrogate to replace the unavailable, position-dependent interference covariance in the MVDR objective and thereby enable data-driven position optimization without jammer-side priors. We established a local approximation bound for the mismatch between the surrogate and the true objective, and further compared the proposed local-anchor surrogate with a natural historical-average surrogate and proved that the latter introduces a non-vanishing geometric bias in the MVDR cost. Building on this theoretical framework, we proposed the low-complexity PTRSO algorithm to ensure that surrogate improvements translate reliably into true SINR gains. Simulations demonstrated consistent SINR gains over existing baselines, confirming the effectiveness of the proposed approach.

{\appendices
	\section{Proof of Lemma 1}

	We first establish a scalar inequality that will be used componentwise.
	For $\forall u,v \in \mathbb{R}$,
$
			|e^{j u} - e^{j v}|
			= 2|\sin \frac{u-v}{2}|
			\leq
			|u-v|.
$
	Now for $ \forall \theta \in \mathbb{R}$ and $\forall\mathbf{x},\mathbf{y}\in\mathbb{R}^{N_r}$, we have
$
			\|
			\mathbf{a}(\mathbf{x},\theta)
			-
			\mathbf{a}(\mathbf{y},\theta)
			\|^2
			=
			\sum_{n=1}^{N_r}
			|
			e^{j k x_n \sin\theta}
			-
			e^{j k y_n \sin\theta}
			|^2\leq
			\sum_{n=1}^{N_r}
			k^2 \sin^2\theta\,|x_n-y_n|^2
			=
			k^2 \sin^2\theta
			\|
			\mathbf{x}-\mathbf{y}
			\|^2,
	$
	which completes the proof.
	
	\section{Proof of Lemma 2}

	First, we establish the uniform bounds on channel norms.
	Since each entry of the steering vector has unit magnitude, 
	we have
$
			\|
			\mathbf{h}_0(\mathbf{x})
			\|
			\leq
			\|
			\mathbf{A}(\mathbf{x})
			\|_{\mathrm{F}}
			\|
			\boldsymbol{\alpha}
			\|
			=
			\sqrt{N_r L}
			\|
			\boldsymbol{\alpha}
			\|.
$
	Similarly, for each jammer channel, we can obtain
$\|	\mathbf{g}_i(\mathbf{x})\|
			=
			|\zeta_i|
			\|
			\mathbf{a}(\mathbf{x},\phi_i)
			\|
			=
			|\zeta_i|\sqrt{N_r},
			i=1,\ldots,I.
$
		
		Next, we will prove the Lipschitz continuity of $\mathbf{h}_0(\mathbf{x})$
	and $\mathbf{g}_i(\mathbf{x})$ with respect to $\mathbf{x}$ by using Lemma 1. For
	$\forall\mathbf{x},\mathbf{y}\in\mathbb{R}^{N_r}$, we have
$
			\|
			\mathbf{h}_0(\mathbf{x}) - \mathbf{h}_0(\mathbf{y})
			\|
			\leq
			\|
			\mathbf{A}(\mathbf{x}) - \mathbf{A}(\mathbf{y})
			\|_{\mathrm{F}}
			\|
			\boldsymbol{\alpha}
			\|
			\leq
			k
			\|
			\boldsymbol{\alpha}
			\|
			\sqrt{
				\sum_{\ell=1}^{L}
				\sin^2\theta_\ell
			}\,
			\|
			\mathbf{x}-\mathbf{y}
			\|,
$
and
$\|
			\mathbf{g}_i(\mathbf{x}) - \mathbf{g}_i(\mathbf{y})
			\|\leq
			k|\zeta_i||\sin\phi_i|
			\|
			\mathbf{x}-\mathbf{y}
			\|,i=1,\ldots,I.
$
	In order to obtain the Lipschitz continuity of $\mathbf{R}(\mathbf{x})$, we establish the following inequality for $\forall \mathbf{u},\mathbf{v}\in\mathbb{C}^{N_r}$:
$\|
			\mathbf{u}\mathbf{u}^{H}
			-
			\mathbf{v}\mathbf{v}^{H}
			\|
			\leq
			(
			\|\mathbf{u}\| + \|\mathbf{v}\|
		)
			\|
			\mathbf{u}-\mathbf{v}
		\|.
$
	Applying the above inequality with
	$\mathbf{u}=\mathbf{h}_0(\mathbf{x})$ and $\mathbf{v}=\mathbf{h}_0(\mathbf{y})$ yields
$
			\|
			\mathbf{h}_0(\mathbf{x})\mathbf{h}_0(\mathbf{x})^{H}
			-
			\mathbf{h}_0(\mathbf{y})\mathbf{h}_0(\mathbf{y})^{H}
			\|\leq
		(
			\|\mathbf{h}_0(\mathbf{x})\|
			+
			\|\mathbf{h}_0(\mathbf{y})\|
			)
			\|
			\mathbf{h}_0(\mathbf{x}) - \mathbf{h}_0(\mathbf{y})
			\|\leq
			2 k \sqrt{N_r}\,
			\sqrt{L}\,
			\|
			\boldsymbol{\alpha}
			\|^2
			\sqrt{
				\sum_{\ell=1}^{L}
				\sin^2\theta_\ell
			}
			\|
			\mathbf{x}-\mathbf{y}
			\|.
$
	Similarly,  we can obtain
$
			\|
			\mathbf{g}_i(\mathbf{x})\mathbf{g}_i(\mathbf{x})^{H}
			-
			\mathbf{g}_i(\mathbf{y})\mathbf{g}_i(\mathbf{y})^{H}
			\|
			\leq
			2 k \sqrt{N_r}\,
			|\zeta_i|^2
			|\sin\phi_i|
			\|
			\mathbf{x}-\mathbf{y}
			\|.
	$
	Hence, combining the bounds for the desired-signal and interference terms
	and using the triangle inequality as follows:
$
			\|
			\mathbf{R}(\mathbf{x}) - \mathbf{R}(\mathbf{y})
			\|
			\le
			\sigma_s^2(
			\|\mathbf{h}_0(\mathbf{x})\|+\|\mathbf{h}_0(\mathbf{y})\|
			)
			\|\mathbf{h}_0(\mathbf{x})-\mathbf{h}_0(\mathbf{y})\|   
			+
			\sum_{i=1}^{I}\sigma_i^2
			(
			\|\mathbf{g}_i(\mathbf{x})\|+\|\mathbf{g}_i(\mathbf{y})\|
			)
			\|\mathbf{g}_i(\mathbf{x})-\mathbf{g}_i(\mathbf{y})\|,
$
	 we arrive at (\ref{eq:R-Lip-short}),
	which completes the proof.

\section{Proof of Lemma 3}
For each anchor position $\mathbf{x}_i$ and snapshot index
$t=1,\ldots,T$, we first express the received complex signal in the following real-valued form so that the real-domain covariance estimation theorem can be applied:
\begin{equation}
	\begin{aligned}
		\widetilde{\mathbf{r}}(\mathbf{x}_i,t)
		\triangleq
		\begin{bmatrix}
			\Re\{\mathbf{r}(\mathbf{x}_i,t)\} \\
			\Im\{\mathbf{r}(\mathbf{x}_i,t)\}
		\end{bmatrix}
		\in\mathbb{R}^{2N_r}.
	\end{aligned}
\end{equation}
Let
$
		\boldsymbol{\Sigma}_{\mathrm{R}}
		\triangleq
		\mathbb{E}
		[
		\widetilde{\mathbf{r}}(\mathbf{x}_i,t)
		\widetilde{\mathbf{r}}(\mathbf{x}_i,t)^{T}
		]$ and $
		\boldsymbol{\Sigma}_{\mathrm{R},T}
		\triangleq
		\frac{1}{T}
		\sum_{t=1}^{T}
		\widetilde{\mathbf{r}}(\mathbf{x}_i,t)
		\widetilde{\mathbf{r}}(\mathbf{x}_i,t)^{T}
$
denote the real second-moment matrix and its sample version, respectively.

From the signal model~\eqref{eq:r_t} in Section~II, the desired symbols $s_0(t)$ and
jammer symbols $s_i(t)$, $i=1,\ldots,I$, have finite average power
$\sigma_s^{2} = \mathbb{E}\{|s_0(t)|^{2}\}$ and
$\sigma_i^{2} = \mathbb{E}\{|s_i(t)|^{2}\}$, respectively.
In practical systems, these symbols are generated from finite-energy
modulation formats (e.g., PSK/QAM constellations or truncated Gaussian
codebooks), so their real and imaginary parts are bounded random
variables and hence sub-Gaussian in the sense of the $\psi_2$-norm
\cite{highDim}.
Thus, each coordinate
of $\mathbf{r}(\mathbf{x}_i,t)$ is a finite linear combination of
independent sub-Gaussian scalar random variables, and hence is itself
sub-Gaussian.  Consequently, the real vector
$\widetilde{\mathbf{r}}(\mathbf{x}_i,t)$ is a sub-Gaussian random
vector in $\mathbb{R}^{2N_r}$, and
$\{\widetilde{\mathbf{r}}(\mathbf{x}_i,t)\}_{t=1}^{T}$ are i.i.d.\
copies of this vector.

More precisely, by standard properties of sub-Gaussian random
variables \cite{highDim}, there exists a constant $K\ge 1$ such that
for every $\mathbf{v}\in\mathbb{R}^{2N_r}$,
\begin{equation}
	\begin{aligned}
		\big\|
		\langle
		\widetilde{\mathbf{r}}(\mathbf{x}_i,t),
		\mathbf{v}
		\rangle
		\big\|_{\psi_2}
		\le
		K
		\big\|
		\langle
		\widetilde{\mathbf{r}}(\mathbf{x}_i,t),
		\mathbf{v}
		\rangle
		\big\|_{L_2}.
	\end{aligned}
\end{equation}
Therefore, we can apply
the covariance estimation theorem in \cite{highDim} to the random
vector $\widetilde{\mathbf{r}}(\mathbf{x}_i,t)$ with dimension
$2N_r$ and sample size $T$.  In particular, there must exist a universal constant
$C>0$ (independent of $T$ and $N_r$) such that
\begin{equation}
	\label{realFormBound}
	\begin{aligned}
		\mathbb{E}
		\big\|
		\boldsymbol{\Sigma}_{\mathrm{R},T}
		-
		\boldsymbol{\Sigma}_{\mathrm{R}}
		\big\|
		\le
		C K^2
		\bigg(
		\sqrt{\frac{2N_r}{T}}
		+
		\frac{2N_r}{T}
		\bigg)
		\big\|
		\boldsymbol{\Sigma}_{\mathrm{R}}
		\big\|.
	\end{aligned}
\end{equation}

In order to extend this result to our complex-valued model, we define
a matrix
$
		\mathbf{M}
		\triangleq
		[
		\mathbf{I}_{N_r}
		\ \ {j}\mathbf{I}_{N_r}
		]
		\in\mathbb{C}^{N_r\times 2N_r},
$
so that
$
		\mathbf{r}(\mathbf{x}_i,t)
		=
		\mathbf{M}\,
		\widetilde{\mathbf{r}}(\mathbf{x}_i,t).
$
The complex second-moment matrix and its sample estimate can then be
written as
$
		\mathbf{R}(\mathbf{x}_i)
		\triangleq
		\mathbb{E}
		\big[
		\mathbf{r}(\mathbf{x}_i,t)
		\mathbf{r}(\mathbf{x}_i,t)^{ H}
		\big]
		=
		\mathbf{M}
		\boldsymbol{\Sigma}_{\mathrm{R}}
		\mathbf{M}^{ H},
		\hat{\mathbf{R}}(\mathbf{x}_i)
		\triangleq
		\frac{1}{T}
		\sum_{t=1}^{T}
		\mathbf{r}(\mathbf{x}_i,t)
		\mathbf{r}(\mathbf{x}_i,t)^{ H}
		=
		\mathbf{M}
		\boldsymbol{\Sigma}_{\mathrm{R},T}
		\mathbf{M}^{ H}.
$
Hence
\begin{equation}
	\begin{aligned}
		\hat{\mathbf{R}}(\mathbf{x}_i)
		-
		\mathbf{R}(\mathbf{x}_i)
		&=
		\mathbf{M}
		\big(
		\boldsymbol{\Sigma}_{\mathrm{R},T}
		-
		\boldsymbol{\Sigma}_{\mathrm{R}}
		\big)
		\mathbf{M}^{ H}.
	\end{aligned}
\end{equation}
Using the submultiplicativity of the spectral norm and the fact that
$\mathbf{M}\mathbf{M}^{ H}=2\mathbf{I}_{N_r}$, we obtain
\begin{equation}\label{FirstIneq}
	\begin{aligned}
		\big\|
		\hat{\mathbf{R}}(\mathbf{x}_i)
		-
		\mathbf{R}(\mathbf{x}_i)
		\big\|
		&\le
		\big\|
		\mathbf{M}
		\big\|^2
		\big\|
		\boldsymbol{\Sigma}_{\mathrm{R},T}
		-
		\boldsymbol{\Sigma}_{\mathrm{R}}
		\big\|
		=
		2
		\big\|
		\boldsymbol{\Sigma}_{\mathrm{R},T}
		-
		\boldsymbol{\Sigma}_{\mathrm{R}}
		\big\|.
	\end{aligned}
\end{equation}
Taking expectations on both sides of (\ref{FirstIneq}) and combining with the previous bound (\ref{realFormBound}) yields
\begin{equation}
	\label{ineqWithSig}
	\begin{aligned}
		\mathbb{E}
		\left\lbrace \big\|
		\hat{\mathbf{R}}(\mathbf{x}_i)
		-
		\mathbf{R}(\mathbf{x}_i)
		\big\|\right\rbrace 
		&\le
		2 C K^2
		\bigg(
		\sqrt{\frac{2N_r}{T}}
		+
		\frac{2N_r}{T}
		\bigg)
		\big\|
		\boldsymbol{\Sigma}_{\mathrm{R}}
		\big\|.
	\end{aligned}
\end{equation}

Next, we relate $\big\|\boldsymbol{\Sigma}_{\mathrm{R}}\big\|$ to
$\big\|\mathbf{R}(\mathbf{x}_i)\big\|$.  For any vector
$\mathbf{y}\in\mathbb{R}^{2N_r}$, write
\begin{equation}
	\begin{aligned}
		\mathbf{y}
		&=
		\begin{bmatrix}
			\mathbf{p} \\
			\mathbf{q}
		\end{bmatrix},
		\qquad
		\mathbf{p},\mathbf{q}\in\mathbb{R}^{N_r},
	\end{aligned}
\end{equation}
and define
$
		\mathbf{s}
		\triangleq
		\mathbf{p}
		+
		{j}\mathbf{q}
		\in\mathbb{C}^{N_r}.
$
Then,
$\|\mathbf{s}\|^2=\|\mathbf{p}\|^2+\|\mathbf{q}\|^2
=\|\mathbf{y}\|^2$. Moreover,
\begin{equation}
	\begin{aligned}
		\Re\big(
		\mathbf{s}^{H}
		\mathbf{r}(\mathbf{x}_i,t)
		\big)
		&=
		\mathbf{p}^{T}\Re\{\mathbf{r}(\mathbf{x}_i,t)\}
		+
		\mathbf{q}^{T}\Im\{\mathbf{r}(\mathbf{x}_i,t)\}
		\\
		&=
		\mathbf{y}^{T}
		\widetilde{\mathbf{r}}(\mathbf{x}_i,t).
	\end{aligned}
\end{equation}
By the definition of $\boldsymbol{\Sigma}_{\mathrm R}$, it follows that 
\begin{equation}
	\begin{aligned}
		\mathbf{y}^{T}
		\boldsymbol{\Sigma}_{\mathrm R}
		\mathbf{y}
		&=
		\mathbb{E}
		\left\{
		\big(
		\mathbf{y}^{T}
		\widetilde{\mathbf{r}}(\mathbf{x}_i,t)
		\big)^2
		\right\}
		=
		\mathbb{E}
		\left\{
		\Re\big(
		\mathbf{s}^{H}
		\mathbf{r}(\mathbf{x}_i,t)
		\big)^2
		\right\}
		\\
		&\le
		\mathbb{E}
		\left\{
		\big|
		\mathbf{s}^{H}
		\mathbf{r}(\mathbf{x}_i,t)
		\big|^2
		\right\}
		=
		\mathbf{s}^{H}
		\mathbf{R}(\mathbf{x}_i)
		\mathbf{s},
	\end{aligned}
\end{equation}
where we have used $|\Re(z)|\le|z|$ for any complex scalar $z$. Therefore,
for any nonzero $\mathbf{y}$, the following inequality holds:
\begin{equation}\label{RayleighIneq}
	\begin{aligned}
		\frac{
			\mathbf{y}^{T}
			\boldsymbol{\Sigma}_{\mathrm R}
			\mathbf{y}
		}{
			\|\mathbf{y}\|^2
		}
		&\le
		\frac{
			\mathbf{s}^{H}
			\mathbf{R}(\mathbf{x}_i)
			\mathbf{s}
		}{
			\|\mathbf{s}\|^2
		}.
	\end{aligned}
\end{equation}
Taking the maximum over all $\mathbf{y}\neq\mathbf{0}$ on (\ref{RayleighIneq}) yields
\begin{equation}
	\begin{aligned}
		\big\|
		\boldsymbol{\Sigma}_{\mathrm R}
		\big\|
		&=
		\max_{\mathbf{y}\neq\mathbf{0}}
		\frac{
			\mathbf{y}^{T}
			\boldsymbol{\Sigma}_{\mathrm R}
			\mathbf{y}
		}{
			\|\mathbf{y}\|^2
		}
		\le
		\max_{\mathbf{s}\neq\mathbf{0}}
		\frac{
			\mathbf{s}^{H}
			\mathbf{R}(\mathbf{x}_i)
			\mathbf{s}
		}{
			\|\mathbf{s}\|^2
		}
		=
		\big\|
		\mathbf{R}(\mathbf{x}_i)
		\big\|.
	\end{aligned}
\end{equation}
Substituting this relation into \eqref{ineqWithSig} and absorbing all
numerical constants into the constant $C$, we obtain
\begin{equation}
	\begin{aligned}
		\mathbb{E}
		\left\{
		\big\|
		\hat{\mathbf{R}}(\mathbf{x}_i)
		-
		\mathbf{R}(\mathbf{x}_i)
		\big\|\right\}
		\le
		C K^2
		\bigg(
		\sqrt{\frac{N_r}{T}}
		+
		\frac{N_r}{T}
		\bigg)
		\big\|
		\mathbf{R}(\mathbf{x}_i)
		\big\|,
	\end{aligned}
\end{equation}
which completes the proof.

{{
\section{Proof of Lemma~\ref{lem:surrogate_error_bound}}
\label{app:surrogate_error_bound}

We first prove the high probability bound \eqref{eq:event_Ei_probability}.
By Lemma~3,
$\mathbb E[
\|\boldsymbol\Delta_i\|
]\le\varepsilon(T)$, and since
$\|\boldsymbol\Delta_i\|$ is nonnegative, Markov's
inequality gives
\begin{equation}
	\begin{aligned}
	1-	\mathbb P(\mathcal E_i)
		&=
		\mathbb P
		\left(
		\big\|
		\boldsymbol\Delta_i
		\big\|
		>
		\sqrt{\sigma_n^2\varepsilon(T)}
		\right)
		\\
		&\le
		\frac{
			\mathbb E
			[
			\|
			\boldsymbol\Delta_i
			\|
			]
		}{
			\sqrt{\sigma_n^2\varepsilon(T)}
		}\le
		\frac{\varepsilon(T)}{\sqrt{\sigma_n^2\varepsilon(T)}}
		=
		\sqrt{
			\frac{\varepsilon(T)}{\sigma_n^2}
		},
	\end{aligned}
\end{equation}
which proves~\eqref{eq:event_Ei_probability}.

We next control the inverse of the sample covariance on
$\mathcal E_i$. From the covariance model, all desired-signal and
interference covariance terms are positive semidefinite, while the
noise covariance is $\sigma_n^2\mathbf I$. Hence,
$\mathbf R(\mathbf x)\succeq\sigma_n^2\mathbf I$ for every
$\mathbf x\in\mathcal X$. In particular,
$\lambda_{\min}(\mathbf R(\mathbf x_i))\ge\sigma_n^2$.
Since
$\boldsymbol\Delta_i$ is Hermitian, Weyl's
inequality yields, on $\mathcal E_i$, 
$
		\lambda_{\min}
		(
		\hat{\mathbf R}(\mathbf x_i)
		)\ge
		\lambda_{\min}
		(
		\mathbf R(\mathbf x_i)
		)
		-
		\|
		\boldsymbol\Delta_i
		\|\ge
		\sigma_n^2-\sqrt{\sigma_n^2\varepsilon(T)}.
$
Since $\varepsilon(T)\le\sigma_n^2/4$ implies
$\sqrt{\sigma_n^2\varepsilon(T)}\le\sigma_n^2/2$, it follows that
$
\hat{\mathbf R}(\mathbf x_i)
\succeq
\frac{\sigma_n^2}{2}\mathbf I$ and $
\|
\hat{\mathbf R}(\mathbf x_i)^{-1}
\|
\le
{2}/{\sigma_n^2}.
$
For any $\mathbf x\in\mathcal X$, the triangle inequality and
Lemma~2 give, on $\mathcal E_i$,
\begin{equation}
	\begin{aligned}
		\big\|
		\hat{\mathbf R}(\mathbf x_i)
		-
		\mathbf R(\mathbf x)
		\big\|
		&\le
		\big\|
		\hat{\mathbf R}(\mathbf x_i)
		-
		\mathbf R(\mathbf x_i)
		\big\|+
		\big\|
		\mathbf R(\mathbf x_i)
		-
		\mathbf R(\mathbf x)
		\big\|
		\\
		&\le
		\sqrt{\sigma_n^2\varepsilon(T)}
		+
		L_R\|\mathbf x-\mathbf x_i\|.
	\end{aligned}
\end{equation}
Using the inverse identity, 
we can further obtain
\begin{equation}
	\begin{aligned}
		&\big|
		g(\mathbf x)-\hat g(\mathbf x)
		\big|
		=
		\Big|
		\mathbf h_0(\mathbf x)^H
		\big(
		\mathbf R(\mathbf x)^{-1}
		-
		\hat{\mathbf R}(\mathbf x_i)^{-1}
		\big)
		\mathbf h_0(\mathbf x)
		\Big|
		\\
		&\le
		\big\|
		\mathbf h_0(\mathbf x)
		\big\|^2
		\big\|
		\mathbf R(\mathbf x)^{-1}
		\big\|
		\big\|
		\hat{\mathbf R}(\mathbf x_i)
		-
		\mathbf R(\mathbf x)
		\big\|
		\big\|
		\hat{\mathbf R}(\mathbf x_i)^{-1}
		\big\|.
	\end{aligned}
\end{equation}
Since
$\mathbf h_0(\mathbf x)=\mathbf A(\mathbf x)\boldsymbol\alpha$
and
$\|\mathbf A(\mathbf x)\|_{\mathrm F}=\sqrt{N_rL}$, we have
$\|\mathbf h_0(\mathbf x)\|^2
\le N_rL\|\boldsymbol\alpha\|^2$.
Substituting the above bounds yields
\begin{equation}
	\begin{aligned}
		\big|
		g(\mathbf x)-\hat g(\mathbf x)
		\big|
		&\le
		\frac{
			2N_rL\|\boldsymbol\alpha\|^2
		}{
			\sigma_n^4
		}
		\left(
		L_R\|\mathbf x-\mathbf x_i\|
		+
		\sqrt{\sigma_n^2\varepsilon(T)}
		\right)
		\\
		&=
		C_R
		\left(
		L_R\|\mathbf x-\mathbf x_i\|
		+
		\sqrt{\sigma_n^2\varepsilon(T)}
		\right).
	\end{aligned}
\end{equation}
Because the event $\mathcal E_i$ does not depend on the candidate
position $\mathbf x$, the above inequality holds simultaneously for
all $\mathbf x\in\mathcal X$.

Finally, Lemma~3 gives $\varepsilon(T)\to0$ as $T\to\infty$, which implies
$
\mathbb P(\mathcal E_i)
\ge
1-
\sqrt{
	\frac{\varepsilon(T)}{\sigma_n^2}
}
\to1.
$
This completes the proof.

\section{Proof of Theorem~\ref{thm:avg_vs_local}}
\label{app:avg_vs_local}

We first provide a simple lemma on the perturbation of matrix
inverses.

\begin{Lemma}
	\label{UpperboundAppendix}
	Let $\mathbf H_1\succ 0$ and $\mathbf H_2\succ 0$ be two Hermitian
	matrices. Suppose that there exists a constant $\mu>0$ such that
$
			\mu\mathbf I\preceq \mathbf H_j,$
for 
$
			j=1,2.
$
	Then,
$
			\|
			\mathbf H_1^{-1}-\mathbf H_2^{-1}
			\|
			\le
			\frac{1}{\mu^2}
			\|
			\mathbf H_1-\mathbf H_2
			\|.
$
\end{Lemma}

\begin{proof}
	Using the matrix inverse identity, we have
$
			\mathbf H_1^{-1}-\mathbf H_2^{-1}
			=
			\mathbf H_1^{-1}
			\big(
			\mathbf H_2-\mathbf H_1
			\big)
			\mathbf H_2^{-1}.
$
By applying the
	submultiplicative property of the spectral norm yields
$
			\|
			\mathbf H_1^{-1}-\mathbf H_2^{-1}
			\|
			\le
			\|
			\mathbf H_1^{-1}
			\|
			\|
			\mathbf H_2-\mathbf H_1
			\|
			\|
			\mathbf H_2^{-1}
			\|.
$
	Since
	$\mu\mathbf I\preceq\mathbf H_j$ for $j=1,2$, we have
$
			\lambda_{\min}(\mathbf H_j)
			\ge \mu,
			\big\|\mathbf H_j^{-1}\big\|
			=
			\frac{1}{\lambda_{\min}(\mathbf H_j)}
			\le
			\frac{1}{\mu},
			j=1,2.
$
	Substituting the above bounds gives
$
			\|
			\mathbf H_1^{-1}-\mathbf H_2^{-1}
			\|
			\le
			\frac{1}{\mu^2}
			\|
			\mathbf H_1-\mathbf H_2
			\|,
$
	which completes the proof.
\end{proof}

We now prove Theorem~\ref{thm:avg_vs_local}. From the covariance
model~\eqref{eq:Rxx}, all signal covariance terms are positive
semidefinite and the noise covariance is
$\sigma_n^2\mathbf I$ with $\sigma_n^2>0$. Setting
$
	\mu\triangleq\sigma_n^2,
$
we have
$\mathbf R_0
		\succeq \mu\mathbf I$,
	$	\mathbf R_\delta
		\succeq \mu\mathbf I$, and
	$	\widetilde{\mathbf R}
		=
		\frac{1}{M}\sum_{m=1}^{M}\mathbf R(\mathbf x_m)
		\succeq \mu\mathbf I.
$
Therefore, all covariance matrices appearing in the theorem are
positive definite and invertible.
Define the inverse-domain historical bias at the current anchor as
$
		\beta_\star
		\triangleq
		\big\|
		\widetilde{\mathbf R}^{-1}
		-\mathbf R_0^{-1}
		\big\|.
$
Since we have the assumption $\widetilde{\mathbf R}\neq\mathbf R_0$, the inverse mapping is
injective on the set of positive definite matrices, it follows that
$
	\widetilde{\mathbf R}^{-1}
	\neq
	\mathbf R_0^{-1},
$
and hence
$
	\beta_\star>0.
$

Next, define the inverse-domain local-anchor error as
\begin{equation}
	\begin{aligned}\label{defLAE}
		\varepsilon_\delta
		&\triangleq
		\big\|
		\mathbf R_0^{-1}
		-\mathbf R_\delta^{-1}
		\big\|.
	\end{aligned}
\end{equation}
Applying Lemma~\ref{UpperboundAppendix} with
$\mathbf H_1=\mathbf R_0$ and
$\mathbf H_2=\mathbf R_\delta$,  gives
$
		\varepsilon_\delta
		\le
		\frac{1}{\mu^2}
		\big\|
		\mathbf R_\delta-\mathbf R_0
		\big\|.
$
Furthermore, by Lemma~2, we can obtain
\begin{equation}
	\begin{aligned}\label{lemma2app}
		\big\|
		\mathbf R_\delta-\mathbf R_0
		\big\|
		&=
		\big\|
		\mathbf R(\mathbf x_\star+\boldsymbol\delta)
		-\mathbf R(\mathbf x_\star)
		\big\|\le
		L_R\big\|\boldsymbol\delta\big\|.
	\end{aligned}
\end{equation}
Consequently, combining \eqref{defLAE}-\eqref{lemma2app}, we have
\begin{equation}
	\label{eq:local_inverse_error_appendix}
	\begin{aligned}
		\varepsilon_\delta
		&\le
		\frac{L_R}{\mu^2}
		\big\|\boldsymbol\delta\big\|\triangleq
		c_1\big\|\boldsymbol\delta\big\|,
	\end{aligned}
\end{equation}
where
$
		c_1
		\triangleq
		\frac{L_R}{\mu^2}.
$
Thus, for any $\boldsymbol\delta\in\mathbb R^{N_r}$, we can write
$
		\widetilde{\mathbf R}^{-1}
		-\mathbf R_\delta^{-1}
		=
		\big(
		\widetilde{\mathbf R}^{-1}
		-\mathbf R_0^{-1}
		\big)
		+
		\big(
		\mathbf R_0^{-1}
		-\mathbf R_\delta^{-1}
		\big).
$
Applying the reverse triangle inequality yields
\begin{equation}
	\label{eq:reverse_triangle_appendix}
	\begin{aligned}
		\big\|
		\widetilde{\mathbf R}^{-1}
		-\mathbf R_\delta^{-1}
		\big\|
		&\ge
		\big\|
		\widetilde{\mathbf R}^{-1}
		-\mathbf R_0^{-1}
		\big\|
		-
		\big\|
		\mathbf R_0^{-1}
		-\mathbf R_\delta^{-1}
		\big\|
		\\
		&=
		\beta_\star-\varepsilon_\delta.
	\end{aligned}
\end{equation}
Subtracting $\varepsilon_\delta$ from both sides of
\eqref{eq:reverse_triangle_appendix}, and then using
\eqref{eq:local_inverse_error_appendix}, gives
\begin{equation}
	\label{eq:strict_gap_appendix}
	\begin{aligned}
		\big\|
		\widetilde{\mathbf R}^{-1}
		-\mathbf R_\delta^{-1}
		\big\|
		-
		\big\|
		\mathbf R_0^{-1}
		-\mathbf R_\delta^{-1}
		\big\|\ge
		\beta_\star-2\varepsilon_\delta\ge
		\beta_\star
		-2c_1\big\|\boldsymbol\delta\big\|.
	\end{aligned}
\end{equation}
If $c_1=0$, then
\eqref{eq:local_inverse_error_appendix} implies
$\varepsilon_\delta=0$ for every $\boldsymbol\delta$, and the claimed
inequality follows immediately by choosing, for example,
$c_0(\mathbf x_\star)=\beta_\star/2$ and any
$\rho_0(\mathbf x_\star)>0$.

We therefore consider the nontrivial case $c_1>0$, where we can simply choose
$
		c_0(\mathbf x_\star)
		\triangleq
		\frac{\beta_\star}{2},$ and
$
		\rho_0(\mathbf x_\star)
		\triangleq
		\frac{\beta_\star}{4c_1}
		=
		\frac{\mu^2\beta_\star}{4L_R}.
$
Since $\beta_\star> 0$, both constants $c_0(\mathbf x_\star)$ and $\rho_0(\mathbf x_\star)$ are strictly positive. For every
$\boldsymbol\delta$ satisfying
$\|\boldsymbol\delta\|\le\rho_0(\mathbf x_\star)$, we obtain from
\eqref{eq:strict_gap_appendix} that
$
		\|
		\widetilde{\mathbf R}^{-1}
		-\mathbf R_\delta^{-1}
		\|
		-
		\|
		\mathbf R_0^{-1}
		-\mathbf R_\delta^{-1}
		\|
		\ge
		\beta_\star
		-2c_1\rho_0(\mathbf x_\star)
		=
		\frac{\beta_\star}{2}=
		c_0(\mathbf x_\star),
$
which completes the proof.

\section{Proof of Theorem~\ref{thm:sample_gap_stability}}
\label{app:sample_gap_stability}

For brevity, write
$c_0\triangleq c_0(\mathbf x_\star)$,
$\rho_0\triangleq\rho_0(\mathbf x_\star)$, and
$\tau_0\triangleq\tau_0(\mathbf x_\star)$ throughout the proof.

From the covariance model~\eqref{eq:Rxx}, all signal covariance
terms are positive semidefinite, while the noise covariance is
$\sigma_n^2\mathbf I$. Hence, we have 
$\mathbf R_0\succeq\sigma_n^2\mathbf I$ and
$\widetilde{\mathbf R}\succeq\sigma_n^2\mathbf I$. Since
the estimation errors
$\boldsymbol\Delta_\star$ and
$\widetilde{\boldsymbol\Delta}$ are Hermitian, Weyl's inequality gives the following inequality on the event $\mathcal E$:
$
		\lambda_{\min}
		(
		\hat{\mathbf R}(\mathbf x_\star)
		)
		\ge
		\lambda_{\min}(\mathbf R_0)
		-
		\|
		\boldsymbol\Delta_\star
		\|
		\ge
		\sigma_n^2-\tau_0
		\ge
		{\sigma_n^2}/{2}$ and $
		\lambda_{\min}
		(
		\hat{\mathbf R}_{\mathrm{avg}}
		)
		\ge
		\lambda_{\min}
		(
		\widetilde{\mathbf R}
		)
		-
		\|
		\widetilde{\boldsymbol\Delta}
		\|
		\ge
		{\sigma_n^2}/{2}.
$
Thus,
$\hat{\mathbf R}(\mathbf x_\star)
\succeq\frac{\sigma_n^2}{2}\mathbf I$ and
$\hat{\mathbf R}_{\mathrm{avg}}
\succeq\frac{\sigma_n^2}{2}\mathbf I$ on $\mathcal E$.
We further define
$\eta_{\mathrm{loc}}
\triangleq
\|
\hat{\mathbf R}(\mathbf x_\star)^{-1}
-\mathbf R_0^{-1}
\|$,
$\eta_{\mathrm{avg}}
\triangleq
\|
\hat{\mathbf R}_{\mathrm{avg}}^{-1}
-\widetilde{\mathbf R}^{-1}
\|$, and similar to Lemma 4, we can obtain
$
		\eta_{\mathrm{loc}}
		\le
		{2\tau_0}/{\sigma_n^4},
		\eta_{\mathrm{avg}}
		\le
		{2\tau_0}/{\sigma_n^4},
$
and hence
\begin{equation}
	\label{eq:inverse_error_sum_appendix_e}
	\eta_{\mathrm{loc}}
	+
	\eta_{\mathrm{avg}}
	\le
	\frac{4\tau_0}{\sigma_n^4}
	\le
	\frac{c_0}{2},
\end{equation}
where the last inequality follows from
$\tau_0\le c_0\sigma_n^4/8$.

Consider any $\boldsymbol\delta$ satisfying
$\|\boldsymbol\delta\|\le\rho_0$. By the reverse triangle
inequality, Theorem~\ref{thm:avg_vs_local}, and
\eqref{eq:inverse_error_sum_appendix_e}, we obtain
\begin{equation}
	\begin{aligned}
		\big\|
		\hat{\mathbf R}_{\mathrm{avg}}^{-1}
		-&\mathbf R_\delta^{-1}
		\big\|
		\ge
		\big\|
		\widetilde{\mathbf R}^{-1}
		-\mathbf R_\delta^{-1}
		\big\|
		-\eta_{\mathrm{avg}}
		\\
		&\ge
		\big\|
		\mathbf R_0^{-1}
		-\mathbf R_\delta^{-1}
		\big\|
		+c_0
		-\eta_{\mathrm{avg}}
		\\
		&\ge
		\big\|
		\hat{\mathbf R}(\mathbf x_\star)^{-1}
		-\mathbf R_\delta^{-1}
		\big\|
		+
		\frac{c_0}{2}.
	\end{aligned}
\end{equation}
The third inequality follows from
$\|
\mathbf R_0^{-1}-\mathbf R_\delta^{-1}
\|
\ge
\|
\hat{\mathbf R}(\mathbf x_\star)^{-1}
-\mathbf R_\delta^{-1}
\|
-\eta_{\mathrm{loc}}$.
Since neither $\mathcal E$ nor the bounds on
$\eta_{\mathrm{loc}}$ and $\eta_{\mathrm{avg}}$ depend on
$\boldsymbol\delta$, the above inequality holds simultaneously for
all $\|\boldsymbol\delta\|\le\rho_0$.

It remains to bound the probability of $\mathcal E$. From
$\widetilde{\boldsymbol\Delta}
=\frac{1}{M}\sum_{m=1}^{M}
\boldsymbol\Delta_m$ and the triangle inequality, we have
$
	\mathbb E
	[
	\|
	\widetilde{\boldsymbol\Delta}
	\|]
	\le
	\frac{1}{M}
	\sum_{m=1}^{M}
	\mathbb E[
	\|
	\boldsymbol\Delta_m
	\|
	]
	\le
	\widetilde{\varepsilon}_M(T).
$
Applying the union bound and Markov's inequality gives
$
		\mathbb P(\mathcal E)
		\ge 1-
		\mathbb P
		(
		\|
		\boldsymbol\Delta_\star
		\|
		>
		\tau_0
		)
		-
		\mathbb P
		(
		\|
		\widetilde{\boldsymbol\Delta}
		\|
		>
		\tau_0)
		\ge 1-
		\frac{
			\varepsilon_\star(T)
			+
			\widetilde{\varepsilon}_M(T)
		}{
			\tau_0
		}.
$
Therefore, for every fixed finite $M$, Lemma~3 implies
$\varepsilon_\star(T)\to0$ and
$\varepsilon_m(T)\to0$ for all $m=1,\ldots,M$. It follows that
$\widetilde{\varepsilon}_M(T)
=\frac{1}{M}\sum_{m=1}^{M}\varepsilon_m(T)\to0$.
Since $\tau_0>0$ is independent of $T$, we conclude that
$\mathbb P(\mathcal E)\to1$ as $T\to\infty$. This completes the
proof.}}
	
}

\vspace{-0.0cm}
\bibliography{ref1.bib}

\bibliographystyle{IEEEtran}

\end{document}